\documentclass[11pt,a4paper]{article}

\usepackage{tikz,tikz-3dplot}
\usetikzlibrary{arrows.meta,calc}
\usepackage{authblk}

\usepackage[T1]{fontenc}
\usepackage{lmodern}
\usepackage[margin=27mm]{geometry}
\usepackage{amsmath,amssymb,amsthm,mathtools}
\usepackage{microtype,booktabs,array,enumitem,xcolor}
\usepackage{hyperref}
\hypersetup{colorlinks=true,linkcolor=blue!45!black,citecolor=blue!45!black,
 urlcolor=blue!45!black,pdftitle={Robust Approximation and the Arity Barrier at Width Two},
 bookmarksnumbered=true,bookmarksopen=true,bookmarksopenlevel=1}
\newtheorem{theorem}{Theorem}[section]
\newtheorem{lemma}[theorem]{Lemma}
\newtheorem{proposition}[theorem]{Proposition}
\newtheorem{corollary}[theorem]{Corollary}
\theoremstyle{definition}
\newtheorem{definition}[theorem]{Definition}
\theoremstyle{remark}
\newtheorem{remark}[theorem]{Remark}
\newcommand{\R}{\mathbb R}
\newcommand{\Q}{\mathbb Q}

\newcommand{\I}{\mathcal I}
\newcommand{\B}{\mathcal B}

\newcommand{\LM}{\operatorname{LM}}
\newcommand{\NF}{\operatorname{NF}}
\newcommand{\ml}{\operatorname{ml}}
\newcommand{\Maj}{\operatorname{Maj}}
\newcommand{\Typ}{\operatorname{Typ}}
\newcommand{\aff}{\operatorname{aff}}
\newcommand{\conv}{\operatorname{conv}}
\newcommand{\poly}{\operatorname{poly}}
\newcommand{\bits}[1]{\langle #1\rangle}
\newcommand{\supp}{\operatorname{supp}}

\setlist[enumerate]{itemsep=3pt,topsep=5pt}
\newcommand{\iprod}[2]{\langle #1, #2\rangle}

\newcommand{\pE}{\widetilde{\mathbb E}}

\title{Robust Approximation and the Arity Barrier at Width Two%
\thanks{Supported by the Swiss National Science Foundation projects
no.~200021\_207429/1 ``Ideal Membership Problems and the Bit Complexity
of Sum of Squares Proofs'' and no.~200021\_212929/1 ``Computational
methods for integrality gaps analysis''.}}
\author[1]{Nathan Benedetto Proen\c{c}a}
\author[2]{Kopp\'any Istv\'an Encz}
\author[1]{Monaldo Mastrolilli}
\affil[1]{Scuola Universitaria Professionale della Svizzera Italiana, IDSIA, Lugano, Switzerland}
\affil[2]{Universit\`a della Svizzera Italiana, IDSIA, Lugano, Switzerland}
\date{}
\begin{document}
\maketitle
\begin{abstract}
Zwick's algorithm for Horn satisfiability shows that constraints of
unbounded arity can admit a robust approximation guarantee independent
of the arity. For finite constraint languages, Barto and Kozik proved
that robust approximability is characterized by bounded width. We ask whether
arity-independent robustness persists beyond width one and show that
it fails already at width two. For Majority-closed Boolean linear
constraints of maximum arity $k\ge2$, we give a randomized polynomial-time
algorithm that, without knowing $\varepsilon$, violates an expected
$O(\sqrt{\varepsilon\log k})$ fraction of the constraint weight on
$(1-\varepsilon)$-satisfiable instances. Under the Unique Games
Conjecture (UGC), a matching NP-hardness lower bound of
$\Omega(\sqrt{\varepsilon\log k})$ holds in an explicit parameter
regime. 
Under this assumption, the bounded-width characterization therefore
does not extend uniformly to constraints of unbounded arity.

The algorithm rounds a degree-eight Sum-of-Squares (SoS) relaxation
with a single Gaussian threshold. Since polynomial size alone does not make
SoS solvable in polynomial bit complexity, we construct complete
truncated Gr\"obner bases for the soft Majority ideal and show that,
at every fixed degree $2d$, the relaxation can be optimized to any
rational accuracy in polynomial time with exactly feasible solutions,
and degree-$2d$ SoS proofs can be found after an additive perturbation
at degree at most $4d+10$. The lower bound combines Raghavendra's
gap-to-hardness theorem with an integrality gap on a Gaussian star,
analyzed via the Isaksson--Mossel theorem that parallel halfspaces
maximize the joint membership probability of exchangeable Gaussians.
\end{abstract}

\clearpage
\tableofcontents
\clearpage

\section{Introduction}
\label{sec:introduction}

In the context of Constraint Satisfaction Problems (CSPs), robust satisfiability asks to find, in polynomial time, an assignment satisfying
$1-g(\varepsilon)$ of the constraints whenever the instance has an
assignment satisfying $1-\varepsilon$, where
$g(\varepsilon)\to0$ as $\varepsilon\to0$. For fixed finite constraint
languages, bounded width is exactly the algebraic condition behind
this phenomenon: the same local consistency structure that decides
feasibility also supports robust
approximation~\cite{BartoK14,BartoK16,GuruswamiZ12}. This classification, however,
fixes every relation and hence fixes its arity. It does not answer the
uniform question in which a single constraint may
contain an unbounded number of variables.

There is a compelling positive precedent for asking whether bounded
width continues to guarantee robustness in this larger model. Zwick's
algorithm robustly approximates Horn satisfiability even for clauses of
unbounded arity~\cite{Zwick98}; in the terminology relevant here, this
is the paradigmatic width-one case. Its loss is independent of the
clause arity and tends to zero with the optimum deficit (the quantitative
form is $O(\log\log(1/\varepsilon)/\log(1/\varepsilon))$). Thus width one
can combine three properties that need not coincide once the arity
is unbounded: polynomial-time feasibility, a polynomial-size relaxation,
and a robust guarantee uniform over all arities. This naturally raises
the next question: does an analogous arity-independent guarantee hold
for every constant width, or at least at width two?

We work in a succinct input model: each constraint is a single linear
inequality $T_j-\sum_i a_{ji}x_i\ge0$ with rational coefficients in
binary (Definition~\ref{def:input}). A row of arity $k$ is specified by
$k$ variable indices and $k+1$ binary-encoded rationals; its encoding
length includes the bit lengths of these data. In contrast, an
explicit table of satisfying assignments may contain $2^k$ tuples.
Our lower bounds do not depend on
this choice: the hard instances have constant arity, so they apply
equally when relations are given explicitly. The input model affects
only the upper bound. With explicit relations the basic SDP already
supplies the penalties used by our rounding, while in the succinct
model we replace it by a polynomial-size degree-eight SoS relaxation,
whose exact solution is the subject of Part~\ref{part:foundations}.

\paragraph{The width-two barrier.}
This paper gives a negative answer already at width two. We study
Boolean linear relations preserved by coordinatewise majority. By the
Baker--Pixley theorem they are 2-decomposable~\cite{BakerPixley75}:
their unary and binary projections determine the whole relation, and
feasibility of conjunctions of such relations reduces to 2-SAT. They
therefore form a canonical width-two class~\cite{FederV98} with succinct constraints
of unbounded arity. Here we use the relational-width convention:
after replacing a Majority relation by its unary and binary
projections, consistency on pairs with extension through triples
decides feasibility, that is, width $(2,3)$. Width one refers to
consistency on singleton assignments with extension through a
constraint. These conventions do not bound the arity of the original
relations; see~\cite[Section~2.2]{BartoK16}.
Nevertheless, under UGC and
$\mathrm{NP}\not\subseteq\mathrm{BPP}$, there is no randomized
polynomial-time robust guarantee $g(\varepsilon)$ independent of arity.
Throughout, UGC alone yields the NP-hard promise distinctions;
the additional assumption is used when excluding randomized
algorithms (Remark~\ref{rem:hardness-quantifiers}).
Consequently, the arity-independent phenomenon
exhibited by Zwick at width one does not extend even to the next level
of the width hierarchy. Bounded width continues to explain robust
approximability for each fixed finite template, but width alone does
not control robust approximation uniformly over relations of
unbounded arity.

After establishing that a dependence on arity is necessary,
one may ask for its optimal rate.
The general construction of Barto and Kozik
\cite[Section~4.1, Proposition~4.1]{BartoK16} first replaces relations
of maximum arity $k$ over $D$ by unary and binary relations over
$D^k$. Its subsequent SDP has vectors indexed by the expanded domain
\cite[Section~4.2]{BartoK16}, whose size $|D|^k$ is exponential in
$k$. This concerns the size of that formulation in the uniform input
model; it is not a claim that every robust error constant must grow
exponentially with arity.
For the special case of Boolean linear relations, one may simply
collect all the 2-SAT clauses implied by the original instance
and promote a \(g(\varepsilon)\) robust approximation to 2-SAT
into a \(O(k^2 g(\varepsilon))\) robust approximation.
In particular, the algorithm by \cite{CMM2009} provides a
\(O(k^2 \sqrt{\varepsilon})\) robust approximation on constraints of
arity at most \(k\).

\paragraph{A tight quantitative answer.}
We determine exactly how arity enters. If $k$ is the maximum row
arity, the robust error is, up to universal constant factors,
\[
   \sqrt{\varepsilon\log k}
\]
in the common parameter range of the upper and lower bounds. For
$k\ge2$, Theorem~\ref{thm:algorithm} gives an unconditional randomized
polynomial-time algorithm with expected loss
$O(\sqrt{\varepsilon\log k})$. The algorithm does not require
$\varepsilon$ as input, returns a satisfying assignment when the
instance is satisfiable, and runs in time polynomial in the complete
binary encoding of the rational rows and weights, without an
exponential dependence on $k$.

In the other direction, Corollary~\ref{cor:arity-hardness} proves that
under the UGC, for $k\ge21$ and
\[
 0<\varepsilon<\min\left\{\frac9{64},
       \frac9{4\log(k-1)},\frac{\log(k-1)}{256}\right\},
\]
it is NP-hard to distinguish instances with optimum at least
$1-\varepsilon$ from instances with optimum at most
$1-\frac1{16}\sqrt{\varepsilon\log(k-1)}$. The two bounds therefore
match in their dependence on both the deficit and the arity. The
hardness proof gives finite rationally weighted instances, complete
basic-SDP solutions, and separated promise thresholds; the constants
in the rate are not optimized.

Taking $k=1+\lfloor e^{2/\varepsilon}\rfloor$ makes it UGC-hard to obtain
satisfiability above $1-5/64$ on $(1-\varepsilon)$-satisfiable
instances, for every fixed $0<\varepsilon\le1/20$. No function
$g(\varepsilon)\to0$ can therefore describe a randomized
polynomial-time guarantee uniformly over all arities under UGC and
$\mathrm{NP}\not\subseteq\mathrm{BPP}$. This does not conflict with robust
approximation at any fixed arity: for every fixed $k$, our upper bound
still tends to zero with $\varepsilon$. It identifies arity as the
parameter that separates the finite-template theorem and Zwick's
width-one result from the width-two setting considered here.

This perspective also distinguishes our result from fixed-predicate
rounding theorems. Charikar--Makarychev--Makarychev obtain
$O(\sqrt\varepsilon)$ robustness for Max-2SAT using Gaussian threshold
rounding~\cite{CMM2009}. Brakensiek--Ciardo--Guruswami--Potechin--{\v Z}ivn\'y
prove $O_\Gamma(\sqrt\varepsilon)$ for fixed Boolean promise templates
with Majority polymorphisms~\cite[Theorem~4.1, full version]{BCGPZ2026}.
The constants there may
depend on the fixed template. Here the relation itself is a succinct
part of the input, and the theorem exposes the optimal dependence on
its arity. Generic optimal-CSP
rounding~\cite{Raghavendra08,RaghavendraSteurer2009} does not by itself provide
this uniform algorithm or its bit complexity. We use Raghavendra's
gap-to-hardness theorem \cite{Raghavendra08} for the lower bound and
the Gaussian comparison of Isaksson--Mossel~\cite{IsakssonMossel2012}
in the gap construction.

\paragraph{The matching algorithm and its computational foundation.}
The upper bound uses a single Gaussian threshold assignment and a
degree-eight sum-of-squares relaxation. Unary and binary obstructions
to a row receive pseudoexpectation at most the row deficit. After
strongly biased literals are separated, the remaining conflict graph
of a row is a star. Conditioning on its center controls all conflicts
jointly through a Gaussian maximum; this replaces a union bound over
the binary description and produces precisely the factor
$\sqrt{\log k}$. At $d=4$, optimization gives a rational functional
whose normalized objective deficit is at most
$\varepsilon+\eta$, where $\eta>0$ is the requested optimization
accuracy. Finite-precision sampling contributes an arbitrarily small
$\nu>0$, and the explicit loss is
$128\sqrt{(\varepsilon+\eta)(1+\log k)}+\nu$.

Running this relaxation in polynomial time is a substantive part of
the result. Polynomial matrix dimension alone does not guarantee an
exactly feasible rational moment point with polynomial bit
complexity~\cite{ODonnell17}; see~\cite{raghavendra_weitz2017} for
sufficient conditions ensuring polynomial bit complexity of SoS proofs.
We overcome this by constructing, at every fixed degree, its complete
truncated reduced Gr\"obner basis and a rational basis of all truncated
vanishing identities, and by deriving both signs of the structured
basis elements in SoS. The hard Majority Gr\"obner structure is due to
Mastrolilli~\cite{Mastrolilli2021}; the soft extension and the required
derivations use clause extraction, implication paths, Boolean
refutations, and squared activation indicators.

The primal and dual criteria of~\cite{Mastrolilli2026} then yield
Theorem~\ref{thm:main}. At every fixed moment degree $2d$, the
strengthened body $K_{2d}$ admits polynomial-time rational optimization
to any positive rational additive accuracy, while remaining exactly
feasible. Degree-$2d$ SoS proofs can be found after an additive
perturbation at degree at most $4d+10$. At the rounding level,
\[
   \pi_8(S_{26})\subseteq K_8\subseteq S_8.
\]
The algorithm optimizes directly over degree-eight moments and does
not require constructing a degree-26 extension. The Gr\"obner basis
supplies the affine identities needed on the primal side; two-sided
SoS derivability supplies the transfer used on the dual side.

Taken together, the two halves of the paper answer both parts of the
question suggested by the width-one precedent. The algebraic half
shows that the relevant fixed-level relaxation can indeed be computed
in the succinct, unbounded-arity input model. The approximation half
shows exactly what its moments can deliver, and the matching UGC lower
bound proves that the resulting dependence on arity is unavoidable
for every randomized polynomial-time algorithm under the stated
hardness assumptions. Thus the failure of uniform
robustness at width two is not caused by an inability to solve the
relaxation: it is the sharp approximability boundary of the problem.

\paragraph{Organization.}
The introduction is common to both parts of the paper.
Part~\ref{part:approximation} presents the main approximation result
first. Section~\ref{sec:setup} defines the common soft Majority model,
and Section~\ref{sec:computational-preface} states the computational
input that supplies the degree-eight rational pseudoexpectation.
Section~\ref{sec:normalization} normalizes the rows and proves the
low-degree conflict penalties. Section~\ref{sec:rounding} gives the
Gaussian threshold algorithm, its joint star analysis, and its
finite-precision implementation. Section~\ref{sec:hardness} constructs
the Gaussian basic-SDP gap, discretizes it into a finite rational
instance, and derives the UGC lower bound and the arity-independent
barrier. Finally, Section~\ref{sec:combined} states the single theorem
combining primal--dual tractability, the upper bound, and hardness.

Part~\ref{part:foundations} proves the computational input used in
Part~\ref{part:approximation}. Section~\ref{sec:hard} develops the
hard-relation Gr\"obner structure, and Section~\ref{sec:softbasis}
proves the complete truncated basis construction for the soft system.
Section~\ref{sec:sos} gives the two-sided SoS simulation.
Section~\ref{sec:blackboxes} states the dual criterion and transfer
theorem used as black boxes. Section~\ref{sec:primal} states and applies
the primal black box to exact rational moment optimization, and then
combines both routes in the tractability theorem invoked in the first part.

\paragraph{Declaration on the use of AI tools}
The results of this paper, their proofs, and the overall approach are the
authors' own work. AI-based tools were used in a supporting role: to improve
the English exposition, to help check the statements and proofs, and to assist
in drafting this manuscript. All AI-assisted material was verified by the authors, who
take full responsibility for the correctness, originality, and presentation of
the paper.

\clearpage
\part{Matching robust approximation and hardness}
\label{part:approximation}

\section{Soft Majority systems and conventions}
\label{sec:setup}
We define the common input model and the soft system used by both
parts. We also record the algebraic conventions needed to state the
computational input precisely; its construction is proved in
Part~\ref{part:foundations}.

\begin{definition}\label{def:input}
For $a,b,c\in\{0,1\}^n$, their coordinatewise majority is
$\Maj(a,b,c)_i=1$ if and only if $a_i+b_i+c_i\ge2$.
A relation $R\subseteq\{0,1\}^n$ is \emph{Majority-closed} if
$\Maj(a,b,c)\in R$ for every $a,b,c\in R$; the empty relation is
allowed. A rational Boolean linear row
\[
 P_j(x)=T_j-\sum_{i=1}^n a_{ji}x_i\ge0
\]
is Majority-closed when $R_j=\{x\in\{0,1\}^n:P_j(x)\ge0\}$ is.
The input is a list of such rows, each Majority-closed, with rational
coefficients encoded in binary and dimensions $n\ge1$, $m\ge1$.
Each constraint comes equipped with a nonnegative rational weight
$w_j$, and these weights are normalized: $\sum_j w_j = 1$.
The algebraic constructions below also allow $m=0$, independently
of this normalization convention for weighted instances. Arity
is unrestricted. Majority-closedness of a row can be decided in
polynomial time (Remark~\ref{rem:recognition}), and our algorithm
checks it first.
\end{definition}

Choose rational numbers satisfying
\begin{equation}\label{eq:M}
 M_j\ge\max\left\{0,\sum_{i:a_{ji}>0}a_{ji}-T_j\right\}.
\end{equation}
They may be given as input or chosen equal to the right side. The soft
system, in $N=n+m$ variables $z=(x,y)$, is
\begin{equation}\label{eq:soft}
 \widehat P_j(x,y_j)=P_j(x)+M_j(1-y_j)\ge0,
 \qquad x_i^2-x_i=0,\quad y_j^2-y_j=0.
\end{equation}
Zero rows may be discarded as inequalities while their indicator
variables are retained. Constant rows are permitted. Condition
\eqref{eq:M} is precisely a sufficient nonnegative big-$M$ condition
making an inactive row valid on the whole Boolean cube. Thus, with
\[
 Y_K=\prod_{j\in K}y_j,\qquad
 F_K=\bigcap_{j\in K}R_j,\qquad F_\varnothing=\{0,1\}^n,
\]
the feasible set is
\begin{equation}\label{eq:S}
 S=\{(x,\mathbf1_H):H\subseteq[m],\ x\in F_H\}.
\end{equation}
It is nonempty, since $y=0$ is always feasible. No additional coupling
constraints on the indicators are imposed. Such constraints would
invalidate a central decomposition below.
The rational vanishing ideal is
\[
 \I(S)=\{f\in\Q[x,y]:f(s)=0\text{ for all }s\in S\}.
\]
Over $\R$ the truncated identity spaces are the real spans of their
rational versions. Let $I_{\mathrm{bool}}=\langle z_i^2-z_i:i\in[N]\rangle$.

Fix \emph{graded lexicographic} order with variable priority
$x_1>\cdots>x_n>y_1>\cdots>y_m$. The same priority restricted to
$x$ is used for every hard subsystem. Let $G_t$ be the elements of
degree at most $t$ in the \emph{full reduced} basis of $\I(S)$.
This definition does not mean a degree-bounded run of Buchberger's
algorithm. Degrees throughout are ordinary total degrees before
Boolean multilinearization. Write
\[
 W_s=\binom{N+s}{s},\qquad
 E_t(m)=\sum_{k=0}^{\min(t,m)}\binom{m}{k}.
\]
The input length $L$ includes dimensions, coefficients, and the chosen
$M_j$. A complexity bound $\poly(W_s,L)$ is a uniform Turing bit bound;
it is polynomial in the input size for each fixed $s$.

A degree-$D$ SoS proof of $h\ge0$ from \eqref{eq:soft} is an identity
\begin{equation}\label{eq:proof}
 h=\sigma_0+\sum_j\sigma_j\widehat P_j
       +\sum_{i=1}^N q_i(z)(z_i^2-z_i),
\end{equation}
where each $\sigma_j$ is a sum of squares and every displayed summand
has degree at most $D$. Rational PSD Gram matrices are the certificate
encoding; arbitrary rational ideal multipliers $q_i$ are allowed.
All certificates constructed in Section~\ref{sec:sos} can also be
written as sums of squares of rational polynomials.

\section{From rational moment optimization to robust approximation}
\label{sec:computational-preface}
We define the weighted instance by assigning to the rows nonnegative
rational weights $w_j$ summing to one.
Discard zero-weight rows, and let $k\ge1$ bound the number of
variables in a row.
Constant rows are allowed. Write
\[
 \operatorname{OPT}_{\rm int}
 =\max_{x\in\{0,1\}^n}\sum_jw_j\mathbf1_{\{P_j(x)\ge0\}},
 \qquad F_w(y)=\sum_jw_jy_j.
\]
The big-$M$ condition~\eqref{eq:M} implies that maximizing $F_w$ on
the Boolean soft feasible set gives exactly $\operatorname{OPT}_{\rm int}$:
set $y_j$ to the satisfaction indicator of row $j$. A general feasible
point may have $y_j=0$ even when that row is satisfied.
All logarithms in this part are natural unless a base is displayed.

Let $S_8$ denote the ordinary degree-eight soft-system moment body:
its elements are normalized degree-eight functionals on the Boolean
quotient that are nonnegative on $q^2$ for $\deg q\le4$ and on
$q^2\widehat P_j$ whenever $\deg(q^2\widehat P_j)\le8$.
Equivalently, its moment matrix has order four and the localizing
matrix of every nonconstant linear row has order three. The complete
matrix definition at arbitrary degree appears in
Section~\ref{sec:primal}.

\begin{theorem}[Computational input to rounding]
\label{thm:rounding-input}
For rational $\eta>0$, there is a polynomial-time Turing algorithm
that computes a rational degree-eight pseudoexpectation $\pE$ exactly
feasible for
\begin{equation}\label{eq:rounding-K8}
 K_8=\{u\in S_8:L_u(h)=0\text{ for all }h\in\I(S)_{\le8}\}
\end{equation}
such that
\begin{equation}\label{eq:rounding-opt}
 \pE[F_w]\ge\max_{u\in K_8}L_u(F_w)-\eta
             \ge\operatorname{OPT}_{\rm int}-\eta.
\end{equation}
Time and output bit length are bounded by
\begin{equation}\label{eq:rounding-time}
 \poly\!\left(\binom{n+m+8}{8},L,\bits w,\bits\eta\right).
\end{equation}
Its deficits satisfy
\begin{equation}\label{eq:rounding-deficits}
 e_j=\pE[1-y_j]\in[0,1],\qquad
 \delta=\sum_jw_je_j=1-\pE[F_w]\le\varepsilon+\eta
\end{equation}
whenever $\operatorname{OPT}_{\rm int}\ge1-\varepsilon$.
\end{theorem}
\begin{proof}
Apply Theorem~\ref{thm:main}(iii) with $d=4$ and objective $F_w$.
Every feasible evaluation belongs to $K_8$. Positivity on the squares
$y_j^2$ and $(1-y_j)^2$, and the Boolean identities, give the deficit bounds.
\end{proof}

The computed moment matrix has order four, and each linear soft row
has a localizing matrix of order three. The functional lies in $S_8$, the body used by the rounding
analysis below. Its relation to higher levels of the hierarchy is not
needed in this part. For completeness we record that
Theorem~\ref{thm:main}(iv) also gives
\begin{equation}\label{eq:rounding-levels}
 \pi_8(S_{26})\subseteq K_8\subseteq S_8,\qquad
 \max_{S_{26}}F_w\le\max_{K_8}F_w\le\max_{S_8}F_w,
\end{equation}
where $S_{26}$ is the degree-26 soft moment body of
Section~\ref{sec:primal}. The algorithm optimizes only over $K_8$; it
neither optimizes over $S_{26}$ nor extends the computed moments to
that body. Feasibility is exact, and $\eta$ affects only the objective.

We will show that the resulting expected
loss is at most
\begin{equation}\label{eq:announced-approximation}
 \min\{1,\,128\sqrt{(\varepsilon+\eta)\log(ek)}+\nu\},
\end{equation}
where $\nu>0$ is a separate, arbitrarily prescribed sampling tolerance.

\section{Literal normalization and low-degree penalties}
\label{sec:normalization}
\subsection{Unary obstructions and binary conflicts}
After collecting coefficients of repeated variables, a signed row can
be written as
\begin{equation}\label{eq:literal-row}
 P_j=T_j-\sum_{\ell\in\Lambda_j}a_{j\ell}\,\ell(x),\qquad a_{j\ell}>0.
\end{equation}
Here \(V_j\subseteq[n]\) is the set of variables occurring in row \(j\),
and \(\Lambda_j\) is its set of \emph{literal occurrences}: one literal
\(\ell\in\{x_i,1-x_i\}\) for each \(i\in V_j\).
Each variable occurs at most once in a row, so the literals and the
variables of a row correspond bijectively and \(|\Lambda_j|=|V_j|\le k\).
Throughout, a subscript \(\ell,\ell'\) or \(\ell''\) ranges over
\(\Lambda_j\), while the letters \(i,i'\) are reserved for variables.
The threshold is
$T_j=T_j^{\text{old}}+\sum_{i:a_{ji}^{\rm old}<0}|a_{ji}^{\rm old}|$ and the
new coefficients are the absolute values of the nonzero old ones.
Complementation commutes with majority, so the normalized relation
is still Majority-closed. Its soft row is
$\widehat P_j=P_j+M_j(1-y_j)$ with
\(M_j\ge\max\{0,\sum_\ell a_{j\ell}-T_j\}\).
Occurrences in different rows refer to the same global variables;
the literal normalization does not make their choices independent.

Let row~\eqref{eq:literal-row} be Majority-closed with \(T_j\ge0\).
Call a singleton \(\{\ell\}\subseteq\Lambda_j\) \emph{forbidden}, or a unary
obstruction, when \(a_{j\ell}>T_j\), and a pair
\(\{\ell,\ell'\}\subseteq\Lambda_j\) \emph{forbidden}, or a binary
obstruction, when \(a_{j\ell}+a_{j\ell'}>T_j\).

\begin{lemma}[Conflicts in a linear Majority relation]
\label{lem:linear-conflicts}

Let row~\eqref{eq:literal-row} be Majority-closed with \(T_j\ge0\).
An assignment satisfies the row if and only if, in every forbidden
singleton and in every forbidden pair, at least one literal is false.
If \(T_j<0\) the relation is empty; if \(T_j\ge\sum_\ell a_{j\ell}\) it is
universal.
The descriptions can be constructed in polynomial bit time.
\end{lemma}
\begin{proof}
In literal coordinates the relation is downward closed. If a minimal
forbidden support $A$ had at least three members, deleting three
different members would give three allowed supports whose majority
is $A$, a contradiction. Thus every minimal forbidden support has
size one or two. Their tests are exactly the displayed inequalities.
The construction uses $O(k^2)$ rational comparisons per row. This is
the downward-closed linear specialization of the 2-decomposability
property in Lemma~\ref{lem:2decomp}~\cite{BakerPixley75}.
\end{proof}

\begin{remark}[Recognizing Majority-closed rows]\label{rem:recognition}
Majority-closedness of a row can be tested directly. Call
$\sum_{\ell\in A}a_{j\ell}$ the weight of a support $A\subseteq\Lambda_j$. A row
\eqref{eq:literal-row} with $T_j\ge0$ is Majority-closed if and only if
every support avoiding its unary and binary obstructions has weight at
most $T_j$: one direction is Lemma~\ref{lem:linear-conflicts}, and
conversely the relation is then defined by unary and binary clauses,
hence Majority-closed by Lemma~\ref{lem:2decomp}. A row with $T_j<0$
defines the empty relation, which is Majority-closed. Among the literals
with $a_{j\ell}\le T_j$, call those with $a_{j\ell}>T_j/2$ heavy and the others
light. Two heavy literals always conflict and two light literals never
do. Since the relation is downward closed in literal coordinates, it
suffices to check the support of all light literals and, for each heavy
literal, that literal together with the light literals compatible with
it. This takes $O(k^2)$ rational comparisons per row.
\end{remark}

The conflict graph has vertex set $\Lambda_j$ and an edge
$(\ell,\ell')$ when $a_{j\ell}+a_{j\ell'}>T_j$.
A useful property of this threshold graph is that every triangle-free
induced subgraph is a star together with isolated vertices. Indeed,
choose a vertex $\ell_c$ of largest coefficient in that induced subgraph.
An edge $(\ell,\ell')$ avoiding $\ell_c$ would force both $(\ell_c,\ell)$
and $(\ell_c,\ell')$, forming a triangle. The center is chosen within the induced subgraph.
Unary obstructions are kept separately.

\subsection{Penalties from the original soft inequalities}
All (in)equalities in this subsection are modulo the Boolean equations; the degrees stated below also bound the equality multiples used in
those congruences.

By Lemma~\ref{lem:linear-conflicts}, a row is violated only when every
literal of some forbidden singleton or pair is true.
What the rounding needs is therefore a bound on the pseudo-correlation of
the literals in a forbidden set, that is, on the value the degree-eight
functional assigns to their product.
The next lemma bounds it by the deficit of the row,
and~\eqref{eq:pair-moments} below turns that bound into a statement about
the corresponding Gram vectors.

\begin{lemma}[Degree-eight conflict penalties]\label{lem:penalties}
For a degree-eight feasible pseudoexpectation of the original soft
system, every forbidden singleton or pair $A\subseteq\Lambda_j$ satisfies
\begin{equation}\label{eq:penalty}
 0\le\pE\left[\prod_{\ell\in A}\ell\right]\le e_j.
\end{equation}
An identically false row has $e_j=1$.
\end{lemma}
\begin{proof}
Suppress $j$, put $L_A=\prod_{\ell\in A}\ell$, and set
$\gamma=\sum_{\ell\in A}a_\ell-T>0$. The identity needed is
\begin{equation}\label{eq:penalty-identity}
 (yL_A)^2\widehat P
 \equiv-\gamma(yL_A)^2
       -\sum_{\ell\notin A}a_\ell(yL_A\ell)^2.
\end{equation}
The big-$M$ term vanishes because $y^2(1-y)\equiv0$.
The localized term on the left has degree $2(|A|+1)+1\le7$;
every square on the right has degree at most $2|A|+4\le8$.
Localizing positivity on the left and positivity of all the squares
force $\pE[yL_A]=0$. Moreover,
\[
 L_A\equiv L_A^2,\qquad
 (1-y)(1-L_A)\equiv\big((1-y)(1-L_A)\big)^2.
\]
The last square has degree $2(|A|+1)\le6$. Its expectation is
$e-\pE[L_A]$, giving~\eqref{eq:penalty}. For $T<0$, take
$A=\varnothing$ and $\gamma=-T$ in~\eqref{eq:penalty-identity};
the degrees are at most four and $\pE[y]=0$.
Boolean reductions never increase these degree bounds.
\end{proof}

Set $s_0=1$ and $s_i=2x_i-1$.
The rational first-order moment matrix with entries $\pE[s_i
s_{i'}],\,\,0\le i,i'\le n$ is positive semidefinite with diagonal one.
Choose unit Gram vectors $v_0,v_1,\ldots,v_n$, and put $b_i=\langle
v_0,v_i\rangle$.
These are indexed by variables.
For a literal \(\ell=x_i\) put \((v_\ell,b_\ell)=(v_i,b_i)\), and for
\(\ell=1-x_i\) put \((v_\ell,b_\ell)=(-v_i,-b_i)\).
For each conflict \((\ell,\ell')\),
\begin{equation}\label{eq:pair-moments}
 \pE[\ell\ell']
 =\frac{1+b_\ell+b_{\ell'}+\langle v_\ell,v_{\ell'}\rangle}4
 =\frac14\langle v_0+v_\ell,v_0+v_{\ell'}\rangle\in[0,e_j].
\end{equation}
The two expressions agree because \(v_0\) is a unit vector.

\section{Gaussian threshold rounding with logarithmic arity dependence}
\label{sec:rounding}
\subsection{The global rounding rule}
Let $\Delta\in(0,1]$ be rational with $\delta\le\Delta$, set
$H=\lceil\log_2(2k)\rceil$, and choose the largest dyadic $s=2^{-q}$
such that $s^2\le\Delta/(128H)$. Then
\begin{equation}\label{eq:rounding-scale}
 \frac{\Delta}{512H}<s^2\le\frac{\Delta}{128H},\qquad
 \log(ek)\le H\le2\log(ek).
\end{equation}
Draw a single standard Gaussian vector $g$ and return
\begin{equation}\label{eq:threshold-rule}
 x_i=\mathbf1_{\{\langle g,v_i\rangle+b_i/s\ge0\}}
 \qquad(i\in[n]).
\end{equation}
This is exactly the Gaussian threshold scheme of
Charikar--Makarychev--Makarychev~\cite{CMM2009}, with its scale adapted
to the arity. Indeed, since $b_i=\langle v_0,v_i\rangle$ and $s>0$,
\[
 \mathbf 1_{\{\langle g,v_i\rangle+b_i/s\ge0\}}
 =
 \mathbf 1_{\{\langle v_i,v_0+s g\rangle\ge0\}}.
\]
Write \(p\coloneqq v_0+sg\) for the rounded point, so that a literal
\(\ell\) is true exactly when \(\langle v_\ell,p\rangle\ge0\).
It is also the form used for Majority in
Brakensiek--Ciardo--Guruswami--Potechin--{\v Z}ivn\'y~\cite[Section~4.1, full version]{BCGPZ2026}. The analysis below bounds
the failure of an entire row.
Opposite literal vectors have opposite scores (for $v_i$ and $-v_i$, the indicator expressions in \eqref{eq:threshold-rule} are exactly the opposite), so the assignments are globally consistent, outside null tie events.

\begin{theorem}[Rounding a feasible degree-eight functional]
\label{thm:matching-rounding}
Given any degree-eight feasible pseudoexpectation of the soft system,
the rule~\eqref{eq:threshold-rule} has expected unsatisfied weight at most
\begin{equation}\label{eq:rounding-bound}
 \min\{1,\,128\sqrt{\Delta\log(ek)}\}.
\end{equation}
The bound holds for every $\Delta\in(0,1]$ with $\delta\le\Delta$,
including when $\delta=0$.
\end{theorem}

\subsection{A joint estimate for all conflicts of a row}

The key estimate is Lemma~\ref{lem:joint-star}: all conflicts of a
row are controlled by one Gaussian window. We first reduce to this
star configuration, charging two exceptional events.

Call row $j$ \emph{bad} if $e_j\ge1/16$ and \emph{good} otherwise.
Bad rows have total weight at most
\[
 \sum_{j:e_j\ge1/16}w_j\le16\sum_jw_je_j=16\delta.
\]
Fix a good row. Retain its literal occurrences with bias
$b_\ell\ge-1/8$ and call the others \emph{excluded}. Every forbidden
singleton is excluded: Lemma~\ref{lem:penalties} gives
$\pE[\ell]\le e_j$, hence $b_\ell=2\pE[\ell]-1<-7/8$.
For an excluded literal, the rounding makes it true only if the
standard normal $Z_\ell=\langle g,v_\ell\rangle$ exceeds $1/(8s)$.
Thus
\[
 \Pr[\text{some excluded literal of row }j\text{ is true}]
 \le k\exp\!\left(-\frac1{128s^2}\right).
\]
The rounding itself is unchanged; this is only a partition of the
failure event. After charging a total expected loss of at most
$16\delta+k\exp(-1/(128s^2))$, Lemma~\ref{lem:linear-conflicts}
reduces every remaining violation to a forbidden pair of retained
literals.

Consider their conflict graph. By~\eqref{eq:pair-moments} and
Lemma~\ref{lem:penalties}, every retained edge satisfies
\begin{equation}\label{eq:r}
 \langle v_\ell,v_{\ell'}\rangle
 =4\pE[\ell\ell']-1-b_\ell-b_{\ell'}
 \le4e_j-1+\frac14<-\frac12.
\end{equation}
A triangle would therefore give three unit vectors whose sum has
squared norm strictly less than $3+2(-3/2)=0$, a contradiction.
Moreover, this is a threshold graph: $\ell\ell'$ is an edge exactly
when $a_{j\ell}+a_{j\ell'}>T_j$. Choose a retained vertex $c$ with
largest coefficient. Any edge $\ell\ell'$ not incident with $c$
would also force the edges $c\ell$ and $c\ell'$, creating a triangle.
Consequently every edge meets $c$: the graph is a star, possibly
with isolated vertices.

If there are no edges, the row is satisfied outside the events already
charged. Otherwise it remains to bound jointly the event that the
center and at least one leaf are both true.

\begin{lemma}[Joint star bound]\label{lem:joint-star}
For the retained star of good row $j$, with deficit $e_j$, the probability
that any edge has both endpoints true is at most
\begin{equation}\label{eq:star-bound}
 \frac1{\sqrt{2\pi}}
 \left(\frac{8e_j}{s}+\frac{16}{7}s\log(ek)\right).
\end{equation}
\end{lemma}
\begin{proof}
Let \(c\) denote the literal at the center of the star, and
\(L \subseteq \Lambda_j\) its set of leaves, \(|L| \le k\)
(otherwise the event is empty).
Edge \((c,\ell)\) is violated exactly when \(\iprod{v_c}{p} \ge 0\) and
\(\iprod{v_\ell}{p} \ge 0\).
Since \(c\) is retained, \(b_c=\iprod{v_0}{v_c}\ge-\frac18\); set
\(\alpha \coloneqq 1 + b_c \ge \frac{7}{8}\).
For a vector \(w\) write \(w^\perp \coloneqq w - \iprod{w}{v_c} v_c\) for its
component orthogonal to \(v_c\).
For \(\ell \in L\) write \(v_\ell^\perp = \rho_\ell u_\ell\) with
\(\rho_\ell \ge 0\) and \(u_\ell \perp v_c\) a unit vector or \(0\), and
\(d_\ell = v_c + v_\ell = t_\ell v_c + v_\ell^\perp\) with
\(t_\ell = 1 + \iprod{v_c}{v_\ell}\).
Then
\begin{equation}\label{eq:js-sagitta}
  \rho_\ell^2
  = \| v_\ell - \iprod{v_\ell}{v_c}v_c \|^2
  = 1 - \iprod{v_\ell}{v_c}^2
  = t_\ell(2 - t_\ell) \le 2t_\ell
 \qquad \text{(Figure~\ref{fig:amgm}(a))}.
\end{equation}
By~\eqref{eq:pair-moments} and Lemma~\ref{lem:penalties},
\(\pE[c\,\ell] = \frac{1}{4}\iprod{v_0 + v_c}{v_0 + v_\ell} \le e_j\).
Since \(v_0 + v_\ell = (v_0 - v_c) + d_\ell\) and
\(\iprod{v_0 + v_c}{v_0 - v_c} = 0\),
\begin{equation}\label{eq:js-rhombus}
 4\pE[c\,\ell] = \iprod{v_0 + v_c}{d_\ell} = \alpha\, t_\ell + \iprod{v_0}{v_\ell^\perp},
 \qquad \text{so} \qquad
 \iprod{v_0}{v_\ell^\perp} \le 4e_j - \alpha\, t_\ell
 \qquad \text{(Figure~\ref{fig:amgm}(a))}.
\end{equation}
If edge \((c,\ell)\) is violated, then with \(z_\ell = \iprod{g}{u_\ell}^+\),
\begin{equation}
  \label{eq:js-chain}
  \iprod{v_c}{p}
  \le 8 e_j + \frac{8(sz_\ell)^2}{7}.
\end{equation}
Indeed,
\begin{align*}
 \tfrac{1}{2}\iprod{v_c}{p}
 &\le -\iprod{v_c}{v_\ell}\,\iprod{v_c}{p}
 &&\text{\(\iprod{v_c}{p} \ge 0\), \(-\iprod{v_c}{v_\ell} > \tfrac{1}{2}\)}\\
 &\le \iprod{v_\ell}{p} - \iprod{v_c}{v_\ell}\,\iprod{v_c}{p}
 &&\text{\(\iprod{v_\ell}{p} \ge 0\)}\\
 &= \iprod{v_\ell^\perp}{p}\\
 &= \iprod{v_0}{v_\ell^\perp} + s\iprod{g}{v_\ell^\perp}\\
 &\le 4e_j - \alpha\, t_\ell + s\rho_\ell z_\ell
   &&\text{by~\eqref{eq:js-rhombus};
          \(\iprod{g}{\rho_\ell u_\ell} \le \rho_\ell z_\ell\)}\\
 &\le 4e_j + s\rho_\ell z_\ell - \tfrac{\alpha}{2}\rho_\ell^2
   &&\text{by~\eqref{eq:js-sagitta}, \(\alpha > 0\)}\\
 &= 4e_j + \frac{s^2 z_\ell^2}{2\alpha}
   - \frac{\alpha}{2}\Bigl(\rho_\ell - \frac{s z_\ell}{\alpha}\Bigr)^2\\
 &\le 4e_j + \frac{s^2 z_\ell^2}{2\alpha}.
\end{align*}
The last three lines are the AM--GM step of Figure~\ref{fig:amgm}(b).
Doubling and using \(\alpha \ge \frac{7}{8}\), we
conclude~\eqref{eq:js-chain} for each violated edge \((c,\ell)\).

Since \(u_\ell \perp v_c\) we have \(z_\ell = \iprod{g}{u_{\ell}}^+ =
\iprod{g^\perp}{u_\ell}^+\), so each \(z_\ell\) depends on \(g\) only
through \(g^\perp\).
Hence, with \(\omega(g^\perp) =
8e_j + \frac{8}{7} s^2 \max_{\ell \in L} z_\ell^2\),
\[
 \{\text{some edge is violated}\} \subseteq
 \bigl\{0 \le \iprod{v_c}{p} \le \omega(g^\perp)\bigr\}.
\]
The decomposition \(g = \iprod{g}{v_c} v_c + g^\perp\) is orthogonal, so
\(g^\perp\) and \(\iprod{v_c}{g}\) are independent.
Conditionally on \(g^\perp\) the window \([0, \omega(g^\perp)]\) is therefore
fixed, while \(\iprod{v_c}{p} = b_c + s\iprod{v_c}{g}\) is normal with
variance \(s^2\) and so has density at most \(1/(s\sqrt{2\pi})\).
Thus
\begin{align*}
 \Pr[\text{some edge is violated}]
 &\le \Pr\bigl[0 \le \iprod{v_c}{p} \le \omega(g^\perp)\bigr]
 \\
 &= \mathbb E_{g^\perp}\Bigl[\Pr\bigl[0 \le \iprod{v_c}{p} \le \omega(g^\perp)
   \,\big|\, g^\perp\bigr]\Bigr]\\
 &\le \mathbb E_{g^\perp}\Bigl[\frac{\omega(g^\perp)}{s\sqrt{2\pi}}\Bigr]
  \\
 &= \frac{1}{\sqrt{2\pi}}\Bigl(\frac{8e_j}{s}
   + \frac{8}{7}\, s\, \mathbb E \max_{\ell \in L} z_\ell^2\Bigr).
\end{align*}
Moreover, using the definition of \(z_\ell\) and the Chernoff bound,
\[
  \Pr\left[\max_{\ell \in L} z_\ell^2 > u\right]
  = \Pr\left[\max_{\ell \in L} \iprod{g}{u_\ell} > \sqrt{u}\right]
  \le \sum_{\ell \in L} \Pr\left[\iprod{g}{u_\ell} > \sqrt{u}\right]
  \le k \exp(-\tfrac{u}{2}).
\]
Hence
\begin{align*}
  \mathbb E \max_{\ell \in L} z_\ell^2
  &= \int_{0}^{2 \log k} \Pr\Bigl[\max_{\ell \in L} z_\ell^2 > u\Bigr]du
  + \int_{2 \log k}^{+\infty} \Pr\Bigl[\max_{\ell \in L} z_\ell^2 > u\Bigr]du\\
  &\le 2\log k
    + \int_{2\log k}^\infty k \exp(-u/2)\, du\\
  &= 2\log(ek).
  &&\qedhere
\end{align*}
\end{proof}

\begin{figure}[!ht]
\centering
\pgfmathsetmacro{\tI}{0.25}
\pgfmathsetmacro{\rh}{sqrt(1-(1-\tI)^2)}
\pgfmathsetmacro{\ux}{-0.0946095}
\pgfmathsetmacro{\uy}{0.9955145}
\pgfmathsetmacro{\vox}{0.9987492}
\pgfmathsetmacro{\voz}{0.05}
\tdplotsetmaincoords{66}{124}
\resizebox{\linewidth}{!}{%
\begin{tikzpicture}[
  >={Stealth[length=2.4mm]},
  vec/.style={->,thick},
  lab/.style={font=\small},
  ]
\begin{scope}[scale=3.3,tdplot_main_coords]
  \filldraw[blue!8,draw=blue!40]
     (-1.1,-1.0,-1) -- (1.1,-1.0,-1) -- (1.1,1.35,-1) -- (-1.1,1.35,-1) -- cycle;
\end{scope}
\begin{scope}[scale=3.3]
  \shade[ball color=gray!10,opacity=0.3] (0,0) circle (1);
  \draw[gray!60] (0,0) circle (1);
\end{scope}
\begin{scope}[scale=3.3,tdplot_main_coords]
  \draw[gray!50,densely dotted] (1,0,0) arc[start angle=0,end angle=360,radius=1];
  \coordinate (O)  at (0,0,0);
  \coordinate (N)  at (0,0,1);
  \coordinate (S)  at (0,0,-1);
  \coordinate (V0) at (\vox,0,\voz);
  \coordinate (V0c) at (\vox,0,{\voz+1});
  \coordinate (Vi) at ({\rh*\ux},{\rh*\uy},{-(1-\tI)});
  \coordinate (P)  at ({\rh*\ux},{\rh*\uy},-1);
  \draw[black!55,very thick,domain=-40:62,samples=60,variable=\p]
     plot ({sin(\p)*\ux},{sin(\p)*\uy},{-cos(\p)});
  \draw[red!75!black,semithick,dashed,domain=-0.6:0.98,samples=60,variable=\r]
     plot ({\r*\ux},{\r*\uy},{-1+\r*\r/2});
  \draw[vec] (O) -- (N)  node[lab,above] {\(v_c\)};
  \draw[vec] (O) -- (V0) node[lab,below left] {\(v_0\)};
  \draw[vec,densely dashed,gray!70!black] (O) -- (V0c)
     node[lab,left] {\(v_0 + v_c\)};
  \draw[vec] (O) -- (Vi);
  \node[lab,anchor=south west] at ($(O)!0.8!(Vi)$) {\(v_\ell\)};
  \draw[gray,densely dotted] (O) -- (S);
  \draw[vec,violet] (S) -- (Vi)
     node[lab,midway,above left=-2pt] {\(d_\ell\)};
  \draw[vec,blue!70!black] (S) -- (P)
     node[lab,midway,below=1pt] {\(v_\ell^\perp = \rho_\ell u_\ell\)};
  \draw[vec,green!45!black] (P) -- (Vi);
  \node[lab,green!45!black,anchor=west] at ($(P)!0.45!(Vi)+(0,0.03,0)$)
     {\(t_\ell v_c\)};
  \fill (S)  circle (0.018) node[lab,below left] {\(-v_c\)};
  \fill (Vi) circle (0.018);
  \fill (O)  circle (0.012) node[lab,left] {\(0\)};
\end{scope}
\node[font=\small\bfseries] at (-3.6,3.9) {(a)};
\begin{scope}[xshift=4.9cm,yshift=-2.4cm,x=6.2cm,y=8.2cm]
  \draw[->] (0,0) -- (1.0,0) node[lab,right] {\(\rho_\ell\)};
  \draw[->] (0,0) -- (0,0.6);
  \draw[thick,blue!70!black,domain=0:0.92] plot (\x,{0.5*\x});
  \node[lab,blue!70!black,anchor=south east] at (0.3,0.17) {gain \(s \rho_\ell z_\ell\)};
  \draw[thick,green!45!black,domain=0:0.86,samples=80] plot (\x,{1-sqrt(1-\x*\x)});
  \node[lab,green!45!black,anchor=south] at (0.86,0.49) {cost \(\alpha\, t_\ell\) (sphere)};
  \draw[thick,red!70!black,dashed,domain=0:0.92] plot (\x,{\x*\x/2});
  \node[lab,red!70!black,anchor=west] at (0.93,0.42) {\(\tfrac{\alpha}{2} \rho_\ell^2\) (paraboloid)};
  \draw[densely dotted] (0.5,0) node[lab,below] {\(\rho_\ell = \dfrac{s z_\ell}{\alpha}\)}
      -- (0.5,0.25);
  \draw[<->,thick,orange!85!black] (0.5,0.125) -- (0.5,0.25);
  \node[lab,orange!85!black] (gap) at (0.2,0.36)
     {\(\dfrac{s^2 z_\ell^2}{2\alpha}\)};
  \draw[orange!85!black,thin] (gap.south east) -- (0.495,0.19);
  \node[font=\small\bfseries] at (-0.06,0.66) {(b)};
\end{scope}
\end{tikzpicture}}
\caption{Geometry of the proof of Lemma~\ref{lem:joint-star}, for
example data satisfying its hypotheses.
(a)~The leaf \(v_\ell\) near the antipode \(-v_c\). Its displacement
\(d_\ell = v_c + v_\ell\) splits into the tangential part
\(v_\ell^\perp = \rho_\ell u_\ell\), which lies in the tangent plane at \(-v_c\)
(blue, drawn below the sphere, which rests on it), and the normal part
\(t_\ell v_c\). The grey arc is the great circle through \(-v_c\) and
\(v_\ell\); the dashed red curve is the osculating paraboloid
\(-v_c + \rho u_\ell + \frac{\rho^2}{2} v_c\), which lies below the sphere,
so \(t_\ell \ge \rho_\ell^2/2\)~\eqref{eq:js-sagitta}. The penalty says that
\(d_\ell\) is almost orthogonal to \(v_0 + v_c\)~\eqref{eq:js-rhombus}.
(b)~The AM--GM step in~\eqref{eq:js-chain}, in normalised units
(\(\alpha = 1\), \(s z_\ell = \frac{1}{2}\)). The spherical cost
\(\alpha t_\ell\) is at least \(\frac{\alpha}{2}\rho_\ell^2\), so the
gain minus the cost is bounded above by the quadratic
\(s\rho_\ell z_\ell-\frac{\alpha}{2}\rho_\ell^2\).
This upper bound has maximum \(s^2z_\ell^2/(2\alpha)\) at
\(\rho_\ell=sz_\ell/\alpha\); the orange segment shows this gap.}
\label{fig:amgm}
\end{figure}
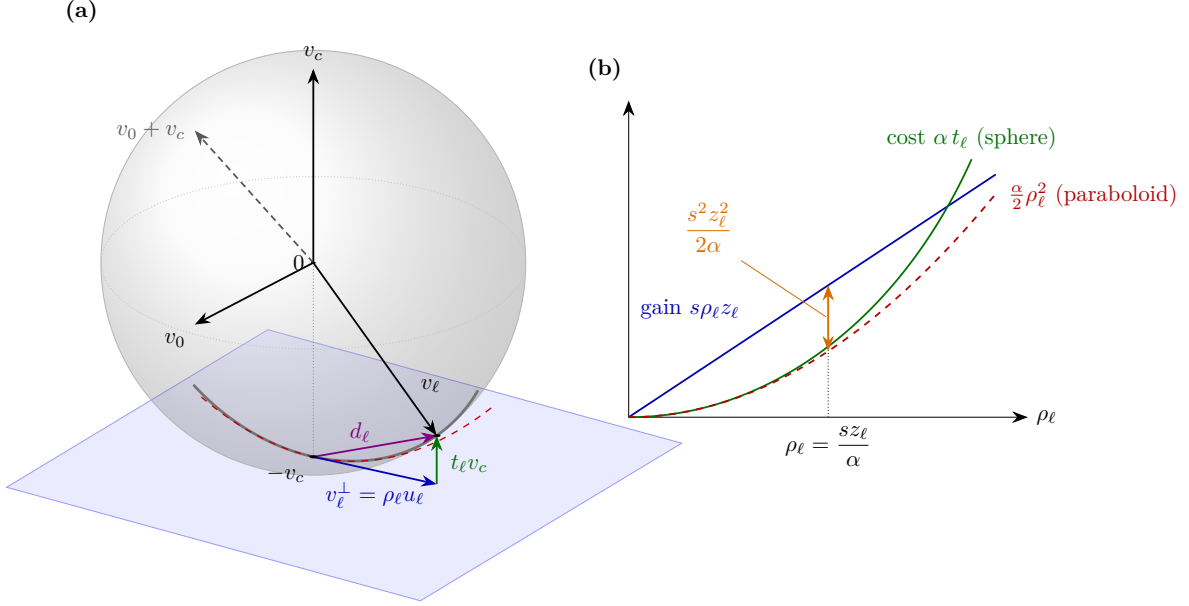

\begin{proof}[Proof of Theorem~\ref{thm:matching-rounding}]
Bad rows have total weight at most $16\delta$. They include all
identically false rows by Lemma~\ref{lem:penalties}. For good rows,
Lemma~\ref{lem:linear-conflicts} accounts for every possible violation.
The exclusion estimate and Lemma~\ref{lem:joint-star} therefore give
\begin{equation}\label{eq:rounding-average}
 \mathbb E[\mathrm{loss}]
 \le16\delta+k e^{-1/(128s^2)}
       +\frac{8\delta}{s\sqrt{2\pi}}
       +\frac{16s\log(ek)}{7\sqrt{2\pi}}.
\end{equation}
For $H\ge1$, the function $e^{-H/u}/u$ is nondecreasing on $(0,1]$.
Thus~\eqref{eq:rounding-scale} implies
\[
 k e^{-1/(128s^2)}\le k e^{-H/\Delta}
 \le k\Delta e^{-H}\le\Delta/e.
\]
Use $\delta\le\Delta$, $\Delta\le\sqrt{\Delta H}$, and both bounds
on $s$. The coefficient of $\sqrt{\Delta H}$ in
\eqref{eq:rounding-average} is at most
\[
 16+e^{-1}+\frac{8\sqrt{512}}{\sqrt{2\pi}}
       +\frac{16}{7\sqrt{256\pi}}<90.
\]
Finally $H\le2\log(ek)$ and $90\sqrt2<128$. The loss is also at most one.
\end{proof}

\subsection{Finite precision and a polynomial-time algorithm}
The cited analyses formulate the rounding using real Gaussian vectors.
The following statement records the finite-bit implementation needed
for our explicit randomized Turing-model guarantee.

\begin{proposition}[Rational implementation]\label{prop:rounding-bits}
From rational moment data, rational $\Delta>0$, and rational
$\nu\in(0,1)$, a randomized Turing algorithm implements the rounding
with expected loss at most
$\min\{1,128\sqrt{\Delta\log(ek)}+\nu\}$.
Time and random-bit usage are polynomial in the input encoding length,
$n,m,k$, $\bits\Delta$, and $\bits\nu$.
For a dyadic tolerance the precision dependence is polynomial in
\(\log_2(1/\nu)+1\).
\end{proposition}
\begin{proof}
$H$ and $q$ are obtained by binary integer and rational comparisons;
$q=O(\bits\Delta+\log H)$. All thresholds $-b_i/s$ are rational
with polynomial bit length. Factor the rational variable Gram matrix
$G$, with entries $G_{ii'}=\pE[s_is_{i'}]$ for $1\le i,i'\le n$,
as $LDL^\top$ by exact rational
elimination, omitting zero pivots. A zero diagonal pivot in a positive
semidefinite residual matrix has a zero row. Bounds on minors give
polynomial bit length for $L,D$. The real matrix $B=L\sqrt D$ has
unit row norms and realizes precisely the Gaussian scores needed.

Set $h=\nu/(10n)$. Choose a rational
$R=O(\sqrt{n\log(2n/\nu)})$ such that
$\Pr(\|g\|_2>R)\le\nu/4$, and compute a rational $\bar B$ with
$\|\bar B-B\|_F\le h/(2R)$. Rational square roots admit bisection
to this accuracy with polynomial bit cost.
For completeness, the following explicit inverse-CDF construction
produces a finite rational $\bar g$ coupled to the independent standard
Gaussian coordinates with $\|\bar g-g\|_2\le h/4$, except with
probability at most $\nu/4$. Put
\[
 \nu_0=\frac{\nu}{32n},\qquad \zeta=\frac{h}{16n},\qquad
 b=\left\lceil\log_2\frac1{\nu_0\zeta}\right\rceil.
\]
For each coordinate use $b$ independent random bits to sample a
uniform integer $J\in\{0,\ldots,2^b-1\}$, and clip $J/2^b$ to
$[\nu_0,1-\nu_0]$. Approximate its standard normal quantile by a
rational within $\zeta$. Write $\Phi$ and $\phi$ for the standard
normal CDF and density. To verify the coupling, let $U$ be uniform
on $[0,1]$ with $J=\lfloor2^bU\rfloor$, and let $g_i=\Phi^{-1}(U)$.
On $U\in[\nu_0,1-\nu_0]$, clipping cannot increase the grid error.
The derivative of $\Phi^{-1}$ on this interval is at most $1/\nu_0$:
for $|x|\ge1$ this follows from the Gaussian tail bound
$\Phi(-|x|)\le\phi(x)/|x|$, and for $|x|\le1$ from
$\phi(x)\ge\phi(1)>\nu_0$. Hence each coordinate error is at most
$2^{-b}/\nu_0+\zeta\le2\zeta$, so
\[
 \|\bar g-g\|_2\le2\sqrt n\,\zeta\le h/4.
\]
There are at most $n$ coordinates, and the excluded uniform tails
have total probability at most $2n\nu_0=\nu/16<\nu/4$.

This construction has polynomial bit cost. The quantiles lie in an
interval of radius $O(\sqrt{\log(1/\nu_0)})$. On that interval the
Gaussian CDF can be evaluated to absolute error at most
$\nu_0\zeta/8$ by its convergent power series, with polynomially many
terms and polynomial working precision in
$\log(1/\nu_0)+\log(1/\zeta)$. Bisection with these error bounds
computes the required quantile approximation: if the comparison is
inconclusive, the midpoint's CDF value is within $\nu_0\zeta/4$ of
the target. Both probabilities then lie in
$[\nu_0/2,1-\nu_0/2]$, where the inverse-CDF derivative is at most
$2/\nu_0$, so the midpoint is within $\zeta/2$ of the target
quantile. Otherwise the enclosing
interval is halved until its width is at most $\zeta$.
The number of random bits per coordinate is
$b=O(\log(n/\nu))$.

On the successful coupling the error in every score is at most $h$,
using $\|g\|_2\le R$ and $\|\bar B_{i,*}\|_2\le2$.
Each exact score is a translated standard normal, even if $G$ is
singular. The probability that any score lies within $h$ of zero is
at most $2nh/\sqrt{2\pi}$. Together with the two exceptional events,
the approximate and exact assignments differ with probability less
than $\nu$. Their losses lie in $[0,1]$, so their expectations differ
by at most $\nu$. Global variables are rounded once and negative
literals are complemented, including at numerical ties.
\end{proof}

\begin{theorem}[Uniform robust approximation]\label{thm:algorithm}
For rational weighted Majority-closed Boolean linear systems of
maximum arity at most $k\ge1$, there is a randomized polynomial-time
algorithm with expected loss at most
\begin{equation}\label{eq:uniform-upper}
 \min\{1,160\sqrt{\varepsilon\log(ek)}\}
\end{equation}
whenever $\operatorname{OPT}_{\rm int}\ge1-\varepsilon$.
It does not need to know $\varepsilon$, and returns a satisfying
assignment when the positive-weight rows are satisfiable.
More precisely, arbitrary rational tolerances $\eta,\nu\in(0,1)$ give
the bound~\eqref{eq:announced-approximation}, in time
\begin{equation}\label{eq:algorithm-complexity}
 \poly\!\left(\binom{n+m+8}{8},L,\bits w,\bits\eta,\bits\nu\right).
\end{equation}
The only optimized moment body is $K_8$.
\end{theorem}
\begin{proof}
First check, by Remark~\ref{rem:recognition}, that every row is
Majority-closed, and reject the input otherwise. Then solve the
conjunction of the rows by the polynomial-size 2-CNF
descriptions of Lemma~\ref{lem:linear-conflicts}; empty rows are
recognized immediately. If satisfiable, return a satisfying assignment.
Otherwise apply Theorem~\ref{thm:rounding-input}. Its rational
deficit $\delta$ is positive: if it were zero,
Theorem~\ref{thm:matching-rounding} with arbitrarily small positive
$\Delta$ would imply that the minimum loss on the finite Boolean cube
is zero, contradicting the failed satisfiability test. Take
$\Delta=\delta\le1$, and use Proposition~\ref{prop:rounding-bits}.
Since $\delta\le\varepsilon+\eta$, this proves the tolerance-dependent
bound and~\eqref{eq:algorithm-complexity}. The encoding length of
$\delta$ is polynomial by Theorem~\ref{thm:rounding-input}.

For a guarantee without a supplied $\varepsilon$, let $w_{\min}>0$
be the minimum row weight and choose $\eta=\nu=w_{\min}/2$.
On an unsatisfiable instance the optimal loss is at least $w_{\min}$,
so $\eta,\nu\le\varepsilon/2$. For $0<\varepsilon\le1$,
\[
 128\sqrt{(\varepsilon+\eta)\log(ek)}+\nu
 \le(128\sqrt{3/2}+1/2)\sqrt{\varepsilon\log(ek)}
 <160\sqrt{\varepsilon\log(ek)}.
\]
All choices have polynomial binary length.
\end{proof}

\subsection{Comparison with fixed-template Majority rounding}
\label{subsec:comparison-rounding}

Our random assignment is the Gaussian threshold rounding of
Charikar--Makarychev--Makarychev~\cite{CMM2009}. Brakensiek--Ciardo--%
Guruswami--Potechin--{\v Z}ivn\'y use the same rule, with scale
\(\sqrt{\varepsilon}\) in place of \(s \in \Theta(\sqrt{\delta/\log(k)})\),
to obtain \(O_\Gamma(\sqrt{\varepsilon})\) robustness for every fixed Boolean
template with Majority polymorphisms~\cite[Section~4 and
Theorem~4.1]{BCGPZ2026}. Their Theorem~4.4 bounds the probability of
one fixed forbidden output, after which a union bound over forbidden
outputs introduces a factor exponential in the constraint arity. Although
harmless for a fixed template, this does not give the best dependency
on the arity \(k\).

The new ingredient here is a joint analysis of an entire linear row.
After strongly biased literals are separated, the retained conflict
graph is a star. Conditional on the Gaussian component \(g^\perp\)
orthogonal to the center vector, each edge
violation provides an upper bound on the same inner product
\(\iprod{v_{c}}{p}\) via~\eqref{eq:js-chain}.
By considering the weakest upper bound, Lemma~\ref{lem:joint-star} obtains
a finer control than a union bound would.

The relaxation being solved is also different, and by necessity.
Without assuming bounded arity, forcing the existence of local distributions
may require exponentially-many constraints.
Here the degree-eight soft moment formulation has polynomial size;
Part~\ref{part:foundations} supplies exactly feasible rational moments, and
Proposition~\ref{prop:rounding-bits} gives a finite-precision
implementation. For fixed arity our relations fall within the scope of
the fixed-template theorem; the contribution here is the explicit
dependence on \(k\) matching the robust approximation algorithm.

\section{Gaussian gaps and UGC hardness}
\label{sec:hardness}
We use $q$ for the number of leaves in a hardness star, so that its
arity is $q+1$. This distinguishes it from the maximum arity $k$ in
the upper bound. All hardness parameters are fixed constants during
the corresponding reduction.

\subsection{The basic SDP and the gap-to-hardness theorem}
For a finite Boolean predicate language, the basic SDP has a local
probability distribution $\mu_C$ on the assignments of each constraint
scope, and a global positive semidefinite matrix $Y$ indexed by $0$
and $(i,a)$ for variables $i$ and labels $a\in\{-1,1\}$. It satisfies
$Y_{00}=1$, $\sum_aY_{ia,0}=1$, and
\[
 Y_{ia,0}=\Pr_{\mu_C}[z_i=a],\qquad
 Y_{ia,i'b}=\Pr_{\mu_C}[z_i=a,z_{i'}=b]
\]
whenever the indicated variables occur in $C$; same-variable entries
encode mutually exclusive labels. Equivalently one can impose the
usual Gram-vector normalization $\sum_a u_{ia}=u_0$,
$\|u_0\|=1$, and $u_{i,+}\perp u_{i,-}$.
The objective is the weighted probability of satisfying the predicates
under the local distributions. Denote its maximum by
$\operatorname{SDP}_{\rm bas}$; this is the standard relaxation
of~\cite{Raghavendra08,RaghavendraSteurer2009}, distinct from the soft
moment bodies $S_8$ and $K_8$.

We use the following consequence of
\cite[Theorem~1.1]{Raghavendra08}: for a fixed finite predicate
language, a finite instance with basic SDP value at least $c$ and
integral value at most $s<c$ implies, under UGC, NP-hardness of
distinguishing integral value at least $c-\xi$ from at most $s+\xi$,
for every fixed $0<\xi<(c-s)/2$. An arbitrarily small additive error
in the gap instance can be absorbed by choosing it smaller than the
hardness slack. We prove all the properties of the gap instance below.

\subsection{An ideal Gaussian star}
Fix $q\ge20$, $0<\rho\le1/4$, and put
$\tau=\sqrt{1-\rho^2}$, $\sigma=\rho\sqrt{\log q}\le1$.
Let \(X_0,X_1,\ldots,X_q\) be independent standard Gaussians in
\(\R^d\), and let \(Y_i=-(\tau X_0+\rho X_i)\).
To a center with \(p\) distinct leaves attach the two predicates below
with equal weights. In the ideal star \(p=q\); after discretization,
\(p\le q\) denotes the number of distinct leaves after merging repeats.
\begin{equation}\label{eq:star-predicates}
 p x_c+\sum_{i=1}^p x_i\le p,
 \qquad p x_c+\sum_{i=1}^p x_i\ge p.
\end{equation}
These are the respective conjunctions of
$(\neg x_c\vee\neg x_i)$ and $(x_c\vee x_i)$, hence are
Majority-closed linear relations. For any Boolean assignment, the
sum of their satisfaction indicators is one plus the indicator that
every leaf is opposite to the center. Their average loss is therefore
one half of the indicator of a failed opposite-label star.

We use the exchangeable-Gaussians theorem of
Isaksson--Mossel~\cite[Theorem~1.2]{IsakssonMossel2012}: for jointly
standard Gaussian vectors with common nonnegative pairwise covariance
$aI$, and measurable sets of prescribed Gaussian measures, the
probability of simultaneous membership is maximized by parallel
halfspaces of those measures. This theorem will be applied only to
the leaves, whose pairwise covariance is $\tau^2 I$.
Write $\phi,\Phi$ for the standard normal density and CDF.

\begin{lemma}[Gaussian star failure]\label{lem:gaussian-failure}
For every measurable \(F:\R^d\to\{0,1\}\),
\begin{equation}\label{eq:gaussian-failure}
 \Pr\big[\exists i\in[q]:F(Y_i)\ne1-F(X_0)\big]
 \ge\frac15\rho\sqrt{\log q}.
\end{equation}
\end{lemma}
\begin{proof}
Let $\beta=\Pr[F(X_0)=1]$. The probability of opposite labels on the
whole star is at most $2\min\{\beta,1-\beta\}$, by its center and
one-leaf marginals. Thus failure is at least $2|\beta-1/2|$.
We may assume $|\beta-1/2|<\sigma/10$, and set $T=\Phi^{-1}(\beta)$.
The inverse CDF is convex above $1/2$; since $\sigma\le1$ and
$\Phi^{-1}(0.6)<0.3$, this gives $|T|<0.3\sigma$.

Failure is also at least one minus the probabilities that all leaf
labels are zero or all are one. Applying the exchangeable-Gaussians
theorem separately to these two events, and reflecting one of the
halfspaces, lower bounds it by
\begin{equation}\label{eq:one-dimensional-crossing}
 \Pr[\tau g+\rho\min_i Z_i<T<\tau g+\rho\max_i Z_i],
\end{equation}
where $g,Z_1,\ldots,Z_q$ are independent standard normals.
Boundary choices have probability zero.

Put $R=(3/5)\sqrt{\log q}$. The Gaussian tail bound
\begin{equation}\label{eq:mills}
 1-\Phi(R)\ge\frac{R\phi(R)}{1+R^2}
\end{equation}
is from \cite{Gordon1941}, and
follows by differentiating the difference of the two sides: its
derivative is $-2\phi(R)/(1+R^2)^2$ and its limit at infinity is zero.
Hence
\[
 \Pr[\min_iZ_i\le-R,\ \max_iZ_i\ge R]
 \ge1-2\exp\!\left(-\frac{qR\phi(R)}{1+R^2}\right)>\frac34.
\]
For clarity, the exponent exceeds $23/10$ already at $q=20$ and
increases with $q$: writing $u=\log q$ and $\alpha=3/5$, its logarithmic
derivative is
$1-\alpha^2/2+1/(2u)-\alpha^2/(1+\alpha^2u)
\ge1-3\alpha^2/2>0$.
The endpoint estimate follows, for example, from
$2.9957<\log20<2.9958$ and the defining formula for $\phi$.

Let $r=\rho R=3\sigma/5$. On this event,~\eqref{eq:one-dimensional-crossing}
holds whenever $\tau g\in(T-r,T+r)$. Since $|T|<r/2$,
$\tau\ge\sqrt{15}/4>0.9$, and $\sigma\le1$, this interval for $g$
lies in $[-1,1]$ and has length $2r/\tau\ge6\sigma/5$.
Independence and the lower bound $\phi(g)\ge\phi(1)$ there give
failure at least $(3/4)(6/5)\phi(1)\sigma>\sigma/5$.
Together with the imbalance case this proves the lemma.
\end{proof}

\subsection{Finite scopes and rational weights}
\begin{lemma}[A quantitative finite discretization]\label{lem:discretization}
Fix $0<\zeta\le1/2$ and $0<a\le1/2$. There is a finite collection
of unit vectors $u_1,\ldots,u_M$ and a probability distribution on
scopes $(c,L)$ with $1\le|L|\le q$ and $c\notin L$, such that
\begin{equation}\label{eq:scope-correlations}
 |\langle u_c,u_i\rangle+\tau|\le\zeta,\qquad
 |\langle u_i,u_{i'}\rangle-\tau^2|\le\zeta
 \quad(i\ne i'\in L),
\end{equation}
and the corresponding paired predicates~\eqref{eq:star-predicates}
have integral value at most $1-\sigma/10+a$.
\end{lemma}
\begin{proof}
Use dimension \(d\ge64(q+1)(q+2)/(a\zeta^2)\).
For the \(q+1\) vectors \(X_0,\widetilde X_i=\tau X_0+\rho X_i\), require
every squared norm divided by \(d\), and every distinct pair inner product
divided by \(d\), to be within \(\zeta/8\) of its expectation.
Each of these quantities has variance at most \(2/d\).
Chebyshev's inequality and a union bound over \((q+1)(q+2)/2\) quantities
show that this event, denoted \(\Typ(\zeta)\), has probability at least
\(1-a\).
On \(\Typ(\zeta)\) the squared norms lie in \([d/2,2d]\); normalizing inner
products introduces total error at most
\(2(\zeta/8)/(1-\zeta/8)<\zeta/2\).

Choose a maximal $h/2$-separated net of the unit sphere with
$0<h\le\zeta/4$, and label each direction by its closest net point,
breaking ties by the smallest index. This is a finite measurable
partition and a covering of radius $h/2$.
Restrict its cells to the annulus with squared norms in \([d/2,2d]\); put
all other points in an extra cell.
Replacing two directions by their net points changes
their inner product by at most $h$.

Condition on \(\Typ(\zeta)\), let \(c\) be the cell of \(X_0\), and let
\(L\) be the set of distinct cells among \(Y_1,\ldots,Y_q\).
The resulting correlations satisfy~\eqref{eq:scope-correlations}.
The center is distinct from every leaf, since a coincident net point
would have inner product one, incompatible with distance at most
$\zeta\le1/2$ from $-\tau$. Repeated leaves can be merged: the
opposite-label event depends only on their distinct labels.
Keep scopes of positive conditional probability.

Extend an assignment to the net cells to a measurable Boolean function
\(F\) on \(\R^d\), using an arbitrary value in the extra cell.
By Lemma~\ref{lem:gaussian-failure}, its unconditioned paired-star
value is at most $1-\sigma/10$. Conditioning on an event of
probability at least $1-a$ changes the expectation of a $[0,1]$-valued
function by at most $a$. This proves the stated integral bound,
uniformly over all assignments.
\end{proof}

For any prescribed $\omega>0$, take $a\le\min\{1/2,\omega/2\}$ and
approximate the finite scope-weight vector in $\ell_1$ by a rational
probability vector within $\omega/2$. Give each predicate in a pair
half its scope weight. Every integral value changes by at most
$\omega/2$, so the integral upper bound is $1-\sigma/10+\omega$.
The SDP construction below satisfies every individual predicate
with the same lower bound, which is therefore unaffected by this
change of weights. The discretization may be very large as a function
of $q,\rho,\omega$; these are fixed parameters for hardness, not input
dimensions of the approximation algorithm.

\subsection{A feasible solution of the full basic SDP}
Consider a \emph{reference star}: three \(\pm1\) random variables
\(z_c,z_1,z_2\) with zero means and correlations
\(\mathbb E[z_cz_1]=\mathbb E[z_cz_2]=-\tau^3\) and
\(\mathbb E[z_1z_2]=\tau^4\).
These are the values of \(\tau^2\langle u_c,u_i\rangle\) and
\(\tau^2\langle u_i,u_{i'}\rangle\) at the ideal correlations \(-\tau\) and
\(\tau^2\) of~\eqref{eq:scope-correlations}, and the construction below
realizes such a star.
Let \(A_1=\{z_1=z_c\}\) and \(A_2=\{z_2=z_c\}\) be the events that a leaf
agrees with the centre, which is what a star predicate forbids, and set
\begin{equation}\label{eq:gap-parameters}
\begin{gathered}
 p_0=\Pr[A_1]=\frac{1-\tau^3}{2},\qquad
 I_0=\Pr[A_1\cap A_2]=\frac{1-2\tau^3+\tau^4}{4},\\
 \lambda=\frac{p_0^2}{I_0},\qquad
 \chi=\frac{I_0}{p_0}=\Pr[A_2\mid A_1].
\end{gathered}
\end{equation}
Thus \(\chi\) is the conditional probability that a second leaf agrees, and
\(\lambda\) is the unique scaling for which \((p_0/\lambda,I_0/\lambda)\) is
the moment pair \((\chi,\chi^2)\) of independent Bernoulli-\(\chi\)
coordinates.
The reference star is never used directly: its role is to fix the point
\((\chi,\chi^2)\) at which Lemma~\ref{lem:correlation-interior} is applied
below, and the discretization is calibrated so that every scope stays close
enough to the star for that lemma to supply a genuine distribution on
subsets of its leaves, which is then the local distribution of the scope.
We will need
\begin{equation}\label{eq:gap-parameter-bounds}
 \frac{\rho^2}{2}\le p_0\le\frac{3\rho^2}{4},\quad
 \frac{\rho^2}{4}\le I_0\le\frac{\rho^2}{2},\quad
 \frac{\rho^2}{2}\le\lambda\le\frac{9\rho^2}{4}<1,
 \quad\frac13\le\chi\le\frac23.
\end{equation}
For the first bounds use $1-u\ge(1-u)^{3/2}\ge1-3u/2$,
$u=\rho^2$, and
$4I_0=\rho^2+(1-\rho^2)(1-\tau)^2$.
The bounds on $\lambda$ and the lower bound on $\chi$ follow.
Finally $\chi=1-(1-\tau^4)/(4p_0)\le2/3$, since
$1-\tau^4\ge\rho^2$ and $4p_0\le3\rho^2$.

\begin{lemma}[An interior neighborhood of product moments]
\label{lem:correlation-interior}
For $p\ge1$, let
\[
 \operatorname{COR}_p
 =\conv\{(\mathbf1_{i\in A},\mathbf1_{\{i,i'\}\subseteq A})_{i,i<i'}:
                          A\subseteq[p]\}.
\]
If $\chi\in[1/3,2/3]$, the cube of $\ell_\infty$ radius
\begin{equation}\label{eq:interior-radius}
 r_p=\frac{3^{-p}}{p+3\binom p2}
\end{equation}
around $(\chi,\chi^2)_{i,i<i'}$ lies in $\operatorname{COR}_p$.
\end{lemma}
\begin{proof}
The independent Bernoulli-$\chi$ distribution has mass at least
$3^{-p}$ at every subset. For a perturbation $(h_i,h_{ii'})$, adjust
the mass of pair $\{i,i'\}$ by $h_{ii'}$ and that of singleton $\{i\}$
by $h_i-\sum_{i'\ne i}h_{ii'}$. Adjust the empty-set mass by the negative
sum of these adjustments. The changes to first and second moments
are exactly the prescribed perturbation. The sum of absolute values
of all nonempty-set adjustments is at most
$(p+3\binom p2)\|h\|_\infty\le3^{-p}$.
Consequently every mass, including that of the empty set, remains
nonnegative, and the total mass stays one.
\end{proof}

Choose in Lemma~\ref{lem:discretization}
\begin{equation}\label{eq:hardness-zeta}
 0<\zeta\le\min\left\{\frac12,
              \frac23\rho^2\min_{1\le p\le q}r_p\right\}.
\end{equation}
The first bound is the hypothesis of Lemma~\ref{lem:discretization}.
The second is what the transfer below needs: correlations within \(\zeta\)
of the reference star give agreement moments within \(3\zeta/4\) of
\((p_0,I_0)\), and dividing by \(\lambda\ge\rho^2/2\) turns that into
\(3\zeta/(2\rho^2)\), which is at most \(r_p\) precisely when
\(\zeta\le\frac23\rho^2r_p\).
For a scope with $p$ leaves, write
$\beta_{0i}=\langle u_c,u_i\rangle$,
$\beta_{ii'}=\langle u_i,u_{i'}\rangle$, and set
\[
 p_i=\frac{1+\tau^2\beta_{0i}}2,\qquad
 d_{ii'}=\frac{1-\tau^2\beta_{ii'}}2,\qquad
 I_{ii'}=\frac{p_i+p_{i'}-d_{ii'}}2.
\]
By~\eqref{eq:scope-correlations},
$|p_i-p_0|\le\zeta/2$ and $|I_{ii'}-I_0|\le3\zeta/4$.
Since $p_0/\lambda=\chi$ and $I_0/\lambda=\chi^2$,
\eqref{eq:gap-parameter-bounds} and~\eqref{eq:hardness-zeta} imply
\[
 \left\|(p_i/\lambda,I_{ii'}/\lambda)_{i,i<i'}
                   -(\chi,\chi^2)_{i,i<i'}\right\|_\infty
 \le\frac{3\zeta}{2\rho^2}\le r_p.
\]
Lemma~\ref{lem:correlation-interior} supplies a distribution on subsets
$T$ of the leaves with first moments $p_i/\lambda$ and second moments
$I_{ii'}/\lambda$. Independently draw an unbiased sign $z_c$.
With probability $1-\lambda$ let $T=\varnothing$; otherwise use that
subset distribution. Set $z_i=z_c$ for $i\in T$ and $z_i=-z_c$
otherwise. This local distribution has zero means and correlations
\begin{equation}\label{eq:local-correlations}
 \mathbb E[z_cz_i]=2p_i-1=\tau^2\beta_{0i},\qquad
 \mathbb E[z_iz_{i'}]=1-2p_i-2p_{i'}+4I_{ii'}=\tau^2\beta_{ii'}.
\end{equation}
All violations occur in the branch of probability $\lambda$.
Each of the two predicates additionally requires a particular center
sign to be violated, so its violation probability is at most
$\lambda/2\le9\rho^2/8$.

To check global SDP feasibility, take orthonormal private vectors
$e_1,\ldots,e_M$, set $v_i=\tau u_i\oplus\rho e_i$, and choose a
unit vector $e_0$ orthogonal to all $v_i$. These $v_i$ have norm one,
and for distinct variables $\langle v_i,v_{i'}\rangle
=\tau^2\langle u_i,u_{i'}\rangle$. Define
\begin{equation}\label{eq:hardness-Gram}
 u_{i,a}=\frac{e_0+a v_i}{2},\qquad
 Y_{ia,i'b}=\langle u_{i,a},u_{i',b}\rangle
             =\frac{1+ab\langle v_i,v_{i'}\rangle}{4},
 \quad Y_{ia,0}=\frac12.
\end{equation}
Together with $Y_{00}=1$, this is a positive semidefinite Gram matrix.
It agrees with the local distributions~\eqref{eq:local-correlations}
because both sides are the same affine function of the correlation.
Indeed, for distinct variables \(i,i'\) of a scope and labels
\(a,b\in\{\pm1\}\) we have \(\mathbf1_{\{z_i=a\}}=\tfrac12(1+az_i)\), so
\begin{align*}
 4\Pr[z_i=a,\ z_{i'}=b]
 &=\mathbb E\bigl[(1+az_i)(1+bz_{i'})\bigr]\\
 &=1+ab\,\mathbb E[z_iz_{i'}]
 &&\text{\(\mathbb{E}[z_i] = \mathbb{E}[z_{i'}] = 0\)}\\
 &=1+ab\langle v_i,v_{i'}\rangle
 &&\text{by~\eqref{eq:local-correlations}}\\
 &=4Y_{ia,i'b}
 &&\text{by~\eqref{eq:hardness-Gram}},
\end{align*}
and likewise \(2\Pr[z_i=a]=\mathbb E[1+az_i]=1=2Y_{ia,0}\).
For equal variables it gives diagonal $1/2$ and probability zero
for opposite labels; moreover $u_{i,+}+u_{i,-}=e_0$.
Thus all constraints of the full basic SDP, including consistency
between scopes, hold.

\begin{theorem}[Finite basic-SDP gaps]\label{thm:gap}
For $q\ge20$, $0<\rho\le1/4$, and $\rho\sqrt{\log q}\le1$, and for
every $\omega>0$, there is a finite rationally weighted instance of
Majority-closed Boolean linear constraints of arity at most $q+1$
such that
\begin{equation}\label{eq:gap-values}
 \operatorname{SDP}_{\rm bas}\ge1-\frac98\rho^2,
 \qquad \operatorname{OPT}_{\rm int}
              \le1-\frac1{10}\rho\sqrt{\log q}+\omega.
\end{equation}
It uses only the paired-star predicates~\eqref{eq:star-predicates}
with $1\le p\le q$.
\end{theorem}
\begin{proof}
Choose~\eqref{eq:hardness-zeta}, apply the finite discretization, and
rationalize its weights as described above. The local distributions
and Gram matrix~\eqref{eq:hardness-Gram} prove the SDP bound for every
choice of weights. Lemma~\ref{lem:gaussian-failure} and the
discretization prove the integral bound.
\end{proof}

\subsection{Separated hardness thresholds and the arity barrier}
\begin{corollary}[UGC gap hardness]\label{cor:ugc-gap}
Assume UGC. In addition to the hypotheses of Theorem~\ref{thm:gap},
suppose
\begin{equation}\label{eq:hardness-separation}
 \rho<\frac4{45}\sqrt{\log q},\qquad
 0<\xi<\frac12\left(\frac1{10}\rho\sqrt{\log q}
                                      -\frac98\rho^2\right).
\end{equation}
It is NP-hard to distinguish instances of arity at most $q+1$ with
\[
 \operatorname{OPT}_{\rm int}\ge1-\frac98\rho^2-\xi
 \quad\text{from}\quad
 \operatorname{OPT}_{\rm int}\le1-\frac1{10}\rho\sqrt{\log q}+\xi.
\]
\end{corollary}
\begin{proof}
For fixed $q$ the star predicates form a finite language. Apply the
gap-to-hardness theorem to Theorem~\ref{thm:gap}, choosing $\omega$
and the transfer slack smaller than $\xi$, with their sum at most
$\xi$. Condition~\eqref{eq:hardness-separation} ensures strictly
separated completeness and soundness thresholds. Rational weights
can be removed by replication for fixed hardness parameters.
\end{proof}

\begin{corollary}[Explicit arity-dependent lower bound]
\label{cor:arity-hardness}
Assume UGC. For every fixed $q\ge20$ and
\begin{equation}\label{eq:epsilon-regime}
 0<\varepsilon<
 \min\left\{\frac9{64},\frac9{4\log q},\frac{\log q}{256}\right\},
\end{equation}
it is NP-hard to distinguish Majority-closed linear systems of arity
at most $q+1$ with integral value at least $1-\varepsilon$ from those
with integral value at most
$1-\frac1{16}\sqrt{\varepsilon\log q}$.
\end{corollary}
\begin{proof}
Take $\rho=(2/3)\sqrt\varepsilon$. The first two restrictions give
$\rho<1/4$ and $\rho\sqrt{\log q}<1$. The third gives
$\sqrt\varepsilon<\sqrt{\log q}/16$, so
$\rho<\sqrt{\log q}/24<(4/45)\sqrt{\log q}$.
The completeness loss before slack is $\varepsilon/2$, while the
soundness loss is $\sqrt{\varepsilon\log q}/15$.
Choose positive $\xi$ smaller than $\varepsilon/2$,
$\sqrt{\varepsilon\log q}/240$, and the upper bound
in~\eqref{eq:hardness-separation}. Corollary~\ref{cor:ugc-gap}
then gives the claimed thresholds. They are disjoint because
$\varepsilon<\sqrt{\varepsilon\log q}/16$.
\end{proof}

\begin{corollary}[No arity-independent robustness]\label{cor:no-robustness}
Assume UGC. For every fixed $0<\varepsilon\le1/20$, it is NP-hard,
on the promise of satisfiability at least $1-\varepsilon$, to produce
an assignment satisfying more than a $1-5/64$ fraction of the rows,
even with arity at most $1+\lfloor e^{2/\varepsilon}\rfloor$.
Consequently, assuming also $\mathrm{NP}\not\subseteq\mathrm{BPP}$,
the class over all arities admits no randomized polynomial-time
robust guarantee $g(\varepsilon)\to0$ independent of arity.
\end{corollary}
\begin{proof}
Put $q=\lfloor e^{2/\varepsilon}\rfloor$. Then
$2-\varepsilon^2\le\varepsilon\log q\le2$ for
$0<\varepsilon\le1/20$: use
$\log(1-t)\ge-2t$ for $0\le t\le1/2$ and
$2e^{-2/\varepsilon}\le\varepsilon$ for the lower bound.
These inequalities imply all three restrictions
in~\eqref{eq:epsilon-regime}, including
$256\varepsilon^2<2-\varepsilon^2$.
Moreover
\[
 \frac1{16}\sqrt{\varepsilon\log q}
 \ge\frac1{16}\sqrt{2-\varepsilon^2}
 \ge\frac{\sqrt2-\varepsilon}{16}>\frac5{64}.
\]
Apply Corollary~\ref{cor:arity-hardness}. If an arity-independent
$g(\varepsilon)\to0$ existed, fix a sufficiently small positive
$\varepsilon$ with $g(\varepsilon)<5/64$ and apply its algorithm to
these fixed-arity hard instances.
\end{proof}

\begin{remark}[Quantifiers and randomized guarantees]
\label{rem:hardness-quantifiers}
Each hardness assertion fixes $q,\varepsilon$ and its slack before
the input size grows. The exponential choice of arity in the last
corollary proves a barrier for the union of the languages; it does
not contradict robust approximation for any one fixed finite
language. The hardness statements are NP-hard promise distinctions
under UGC. Interpreting them as exclusions of randomized algorithms
uses the corresponding randomized hardness convention, or
$\mathrm{NP}\not\subseteq\mathrm{BPP}$ in addition to UGC.
An expected approximation strictly better than a fixed soundness
threshold gives a bounded-error distinguisher by independent
repetition and evaluation of the returned assignments. We do not
claim a deterministic derandomization of Theorem~\ref{thm:algorithm}.
\end{remark}

\section{The combined computational and approximation theorem}
\label{sec:combined}
\begin{theorem}[Primal--dual tractability and matching robust error]
\label{thm:combined}
Let the input be rational Majority-closed Boolean linear rows, with
the uncoupled Boolean activation formulation~\eqref{eq:soft}, normalized
nonnegative rational weights, and maximum arity at most $k\ge1$.
The following statements hold.
\begin{enumerate}[label=(\roman*)]
\item For every fixed integer $d\ge1$, $G_{2d}$ and a basis of all
identities through degree $2d$ are constructible in polynomial bit
time. Each $b\in G_{2d}$ has rational two-sided SoS derivations of
degree at most $2\deg b+10$.
Given rational $\alpha>0$ and rational $f$ admitting a real
degree-$2d$ proof of $f\ge0$, a rational proof of
$f+\alpha\ge0$ of degree at most $4d+10$ is computable in time and
output length $\poly(W_{4d+10},L,\bits f,\bits\alpha)$.
\item For rational $f$ of degree at most $2d$ and rational $\eta>0$,
an exactly feasible,
$\eta$-optimal rational point of $K_{2d}$ is computable in time and
output length $\poly(W_{2d},L,\bits f,\bits\eta)$, and
\[
 \pi_{2d}(S_{4d+10})\subseteq K_{2d}\subseteq S_{2d},\qquad
 \max_{s\in S}f(s)\le\max_{S_{4d+10}}f
                    \le\max_{K_{2d}}f\le\max_{S_{2d}}f.
\]
\item Optimizing $K_8$ and applying Gaussian threshold rounding gives
the randomized polynomial-time bounds
\eqref{eq:announced-approximation} and~\eqref{eq:uniform-upper} for
$\operatorname{OPT}_{\rm int}\ge1-\varepsilon$, with bit complexity
\eqref{eq:algorithm-complexity}. Satisfiable instances are solved
exactly. The upper bound is unconditional.
\item Under UGC, for $k\ge21$ and~\eqref{eq:epsilon-regime} with
$q=k-1$, it is NP-hard to distinguish value at least $1-\varepsilon$
from value at most $1-\frac1{16}\sqrt{\varepsilon\log(k-1)}$.
Thus the dependence $\sqrt{\varepsilon\log k}$ is optimal up to
universal factors in the stated regime; no equality of leading
constants is asserted. Assuming also
$\mathrm{NP}\not\subseteq\mathrm{BPP}$, the class of all arities has
no randomized polynomial-time arity-independent robust guarantee, as in
Corollary~\ref{cor:no-robustness}.
\end{enumerate}
Here $W_r=\binom{n+m+r}{r}$, and all polynomial bounds refer to
binary rational encodings; the degrees are fixed parameters.
\end{theorem}
\begin{proof}
Parts (i)--(ii) are Theorem~\ref{thm:main}; part (iii) is
Theorem~\ref{thm:rounding-input} followed by
Theorem~\ref{thm:algorithm}; part (iv) follows from the two arity
corollaries. Since $\log(ek)$ and $\log(k-1)$ differ by universal
factors for $k\ge21$, the bounds match in order. The construction of
$G_{2d}$ supplies primal identities, while its two-sided derivations
supply dual transfer; degree $26$ is a transfer bound, not the moment
degree rounded in part (iii).
\end{proof}

\clearpage
\part{Computational foundations: primal and dual SoS tractability}
\label{part:foundations}

This part proves the computational theorem used in
Section~\ref{sec:computational-preface}. The construction of the
truncated Gr\"obner basis supplies the complete affine identities
needed by the primal algorithm; the two-sided SoS derivations supply
the dual search and transfer arguments. Keeping these roles separate
also keeps distinct the optimized moment degree, the algebraic
truncation degree, and the degree of the transferred derivations.

\section{Hard Majority relations: a complete basis construction}
\label{sec:hard}
We include the proof to avoid importing any unverified structural
claim into the soft argument.

\begin{lemma}[Two-coordinate description; Baker--Pixley
  \cite{BakerPixley75}]\label{lem:2decomp}
A nonempty Majority-closed Boolean relation is determined by its unary
and binary projections. In particular it is the solution set of all
unary and binary clauses forbidding partial assignments with no extension
in the relation. Conversely, every relation defined by unary and binary
clauses is Majority-closed.
\end{lemma}
\begin{proof}
Suppose a tuple $a$ has every unary and binary restriction extendible
to a tuple in $R$. We show by induction that its restriction to any
$k$ coordinates extends to $R$. This is the assumption for $k\le2$.
For $k\ge3$, choose three different coordinates of the $k$-set.
Inductively choose three tuples matching $a$ on the $k$-set with,
respectively, one of those coordinates deleted. Their majority matches
$a$ everywhere on the $k$-set and belongs to $R$. Taking $k=n$ proves
the claim. For $n=1$ only the unary assertion is needed. For the converse,
each unary or binary clause is preserved by majority:
if a majority tuple falsifies a two-literal clause, at least two of the
three tuples falsify each literal, and these two pairs intersect.
\end{proof}

For one linear row, a partial assignment $x_T=a$ with $|T|\le2$ is
forbidden exactly when
\begin{equation}\label{eq:maxslack}
 T_j-\sum_{i\in T}a_{ji}a_i
      -\sum_{i\notin T}\min(a_{ji},0)<0.
\end{equation}
The same test at $T=\varnothing$ detects an empty row relation.
All these tests are exact rational calculations, at most
$1+2n+4\binom n2$ per row. Therefore a polynomial-size 2-CNF
description of every $R_j$, and hence of every $F_K$, is computable.
The row label is retained on every clause: all clauses extracted from
row $j$ share the \emph{same} soft indicator $y_j$.

\begin{lemma}[Effective 2-SAT operations; Krom \cite{Krom67}]\label{lem:2sat}
Feasibility of a 2-CNF formula, and extendibility of any prescribed
assignment to at most two variables, are decidable in polynomial time.
If the formula is infeasible and has no empty clause, its implication
graph contains a literal $u$ with paths $u\leadsto\neg u$ and
$\neg u\leadsto u$, each of length at most $2n-1$.
\end{lemma}
\begin{proof}
The clause $a\vee b$ gives edges $\neg a\to b$ and $\neg b\to a$;
a unary clause $a$ gives $\neg a\to a$. Edges are closed under
contraposition. If a literal and its complement lie in one strongly
connected component, any truth assignment respecting implications is
impossible. Conversely, order the component DAG topologically and
declare a literal true exactly when its component comes later than
that of its complement. Exactly one of each complementary pair is true.
If $a\to b$, $a$ were true and $b$ false, the topological indices,
together with $\neg b\to\neg a$, would give
$c(\neg b)\le c(\neg a)<c(a)\le c(b)<c(\neg b)$, a contradiction.
This constructs a satisfying assignment. Strong components and paths
are computable by graph traversal. Shortest paths are simple, giving
the bound. Prescribed values are added as unary clauses.
\end{proof}

We use the following monic polynomial families:
\begin{align*}
 \B_X&=\{x_i^2-x_i:i\in[n]\},\\
 \mathcal Q&=\{(x_i-\alpha)(x_j-\beta):i<j,\ \alpha,\beta\in\{0,1\}\},\\
 \mathcal L&=\{x_i,x_i-1:i\in[n]\}
     \cup\{x_i-x_j,x_i+x_j-1:i<j\},\\
 \mathcal Z&=\{1\}.
\end{align*}
The explicit unary part of $\mathcal L$ also covers $n=1$.
Allowing signs and the larger linear family
$\pm((\delta-\beta)x_i+(\gamma-\alpha)x_j+\alpha\beta-\gamma\delta)$
also gives a redundant description. Only the monic forms
above are needed. Each has at most four monomials and coefficients
in $\{0,1,-1\}$.

\begin{proposition}[Hard reduced basis; Mastrolilli \cite{Mastrolilli2021}]
\label{prop:hard}
Let $R$ be the solution set of a 2-CNF formula on $n$ variables.
Its reduced graded lexicographic basis $H_R$ can be computed in
polynomial time and is contained in
$\B_X\cup\mathcal Q\cup\mathcal L\cup\mathcal Z$.
It has at most $2n+n^2+1$ elements.
\end{proposition}
\begin{proof}
If $R=\varnothing$, output $\{1\}$. Otherwise determine all forced
values and all identities $x_i=x_j$ or $x_i=1-x_j$ by the unary/binary
extendibility tests of Lemma~\ref{lem:2sat}. Among the unforced
variables, equivalence up to complementation is an equivalence relation.
Choose the variable of smallest priority in each class as representative.
For each forced variable include $x_i-c_i$. For each nonrepresentative
unforced variable include $x_i-x_j$ or $x_i+x_j-1$, where $x_j$ is its
representative. The leading monomial of this linear polynomial is $x_i$.

Let $V$ be the set of representatives, and let $R'$ be the projection
onto $V$. It is Majority-closed. It has no forced variable and no two
distinct variables equal or complementary on all its tuples. Every
pair of coordinates in $R'$ therefore forbids at most one of its four
assignments: two forbidden assignments in a common row or column would
force a variable, and two diagonally placed forbidden assignments would
force equality or complementation. Include $u^2-u$ for each $u\in V$,
and include $(u-a)(v-b)$ whenever $u=1-a,v=1-b$ is the (unique, if any)
forbidden assignment for that pair. Lemma~\ref{lem:2decomp} says these
polynomials describe $R'$ exactly.

These polynomials form a Gr\"obner basis. The leading variables of the
linear polynomials are distinct and occur in no other leading monomial,
so their pairs satisfy the relatively-prime leading-monomial criterion.
The same criterion handles polynomials in disjoint variables. For two
conflicts sharing $u$, write
$q=(u-a)(v-b)$ and $r=(u-c)(w-e)$, where $u,v,w$ are distinct.
Their S-polynomial has the exact representation
\begin{equation}\label{eq:Spair}
 wq-vr=e q-b r+(c-a)(v-b)(w-e).
\end{equation}
If $a=c$, the last term vanishes. Otherwise the last conflict is
entailed: setting $v=1-b,w=1-e$ would force both $u=a$ and $u=c$.
Hence it is in the list. Every product on the right has degree at
most two, strictly smaller than the degree-three least common
multiple $uvw$. This is a standard representation sufficient for
Buchberger's criterion. There cannot be two distinct conflicts on the
same pair, by the preceding projection argument. Finally,
\begin{equation}\label{eq:Sbool}
 v(u^2-u)-u(u-a)(v-b)
      =b(u^2-u)+(a-1)(u-a)(v-b),
\end{equation}
using $a(a-1)=0$; this handles a Boolean polynomial and a conflict
sharing a variable. Pairs of Boolean polynomials have coprime leading
monomials. This exhausts all pairs.

The generated ideal contains all original Boolean equations: substitute
the linear equations into $x_i^2-x_i$ to obtain either zero or a
representative Boolean equation. Its zeros are precisely $R$.
An ideal containing the Boolean ideal is determined by its Boolean
zeros: the Boolean quotient is the product of copies of $\Q$ indexed
by the cube, and an ideal in that product is zero or the full field
in each component. Thus our ideal is $\I(R)$, not just an ideal with
the same real zeros without a radicality argument.

The list is monic and reduced. A linear tail uses only a representative
and a constant. A quadratic tail is linear, and all quadratic leading
monomials are distinct. No tail is divisible by another leading
monomial. The construction uses $O(n^2)$ 2-SAT tests. Its coefficients
are only $0,1,-1$, proving all remaining assertions.
\end{proof}

\section{The soft basis: structure, completeness, and truncation}
\label{sec:softbasis}
Set $H_K=H_{F_K}$. Intersections and projections of Majority-closed
relations are Majority-closed, so the preceding construction applies.

\begin{lemma}[Activation coefficient decomposition]\label{lem:decomp}
For a Boolean multilinear polynomial write uniquely
$f(x,y)=\sum_{K\subseteq[m]}Y_K f_K(x)$, with every $f_K$ multilinear
in $x$. Then
\begin{equation}\label{eq:decomp}
 f\in\I(S)\quad\Longleftrightarrow\quad
 f_K\in\I(F_K)\text{ for every }K.
\end{equation}
\end{lemma}
\begin{proof}
Suppose $f$ vanishes on $S$. Evaluate at $y=0$ to obtain the assertion
for $K=\varnothing$. Inductively, for $x\in F_K$, evaluation at
$y=\mathbf1_K$ gives $0=\sum_{H\subseteq K}f_H(x)$.
For $H\subsetneq K$, the induction hypothesis and $F_K\subseteq F_H$
give $f_H(x)=0$, hence $f_K(x)=0$. Empty $F_K$ cause no difficulty.
Conversely, at $(x,\mathbf1_H)\in S$ only terms with $K\subseteq H$
survive, and each vanishes because $x\in F_H\subseteq F_K$.
\end{proof}

Define the finite, possibly exponentially large list
\begin{equation}\label{eq:U}
 U=\{y_j^2-y_j:j\in[m]\}
      \cup\bigcup_{K\subseteq[m]}\{Y_K h:h\in H_K\}.
\end{equation}
The corresponding ambient structured families are
\[
 \B_Y=\{y_j^2-y_j\}_j\cup\{Y_K b:b\in\B_X\},\qquad
 \mathcal Q_Y=\{Y_Kq:q\in\mathcal Q\},
\]
and analogously $\mathcal L_Y$ and $\mathcal Z_Y$.
These are \emph{ambient candidate families}; their arbitrary members
need not vanish. The list \eqref{eq:U} selects only valid members.

\begin{theorem}[Constructible reduced soft basis]\label{thm:structure}
The reduced basis $G$ of $\I(S)$ is obtained from $U$ by retaining
exactly the elements whose leading monomial is minimal under
divisibility among the leading monomials of $U$.
Consequently, except for $y_j^2-y_j$, every element has form
\begin{equation}\label{eq:structure}
 p=Y_K\tau(x),\qquad
 \tau\in\B_X\cup\mathcal Q\cup\mathcal L\cup\mathcal Z,
 \qquad\deg p=|K|+\deg\tau.
\end{equation}
For every integer $t\ge0$, the correct truncation $G_t$ is computable
from the rational rows in time $\poly(W_t,L)$ (with the convention
that $n,m$ are included in $L$). It has at most $W_t$ elements,
coefficients in $\{0,1,-1\}$, and at most four terms per element.
\end{theorem}
\begin{proof}
Membership of every element of $U$ follows from \eqref{eq:S}, or
directly from Lemma~\ref{lem:decomp}; the Boolean elements vanish too.
For the Gr\"obner property, take a nonzero $f\in\I(S)$.
If $\LM(f)$ has a repeated variable, it is divisible by a Boolean
leading monomial in $U$ (use $H_\varnothing=\B_X$).
Otherwise Boolean multilinearization preserves its leading monomial:
replacing powers by first powers only lowers monomials in a graded
order. Write its multilinearization as in Lemma~\ref{lem:decomp}.
If its leading monomial has $y$-support $K$, it is
$Y_K\LM(f_K)$. Since $f_K\in\I(F_K)$, some $h\in H_K$ has
$\LM(h)\mid\LM(f_K)$. Then $Y_Kh\in U$ supplies the required
leading divisor. Thus the leading ideal of $U$ equals that of
$\I(S)$; division also shows $U$ generates $\I(S)$.

It remains to verify \emph{reducedness}, not just generation.
There are no duplicate leading monomials in $U$: $y$-support
determines $K$, and $H_K$ is reduced. Discard every element whose
leading monomial is properly divisible by another. Call the retained
list $V$. It is still a Gr\"obner basis. Consider a tail monomial
$Y_Kv$ of a retained $Y_Kh$. It is squarefree, except possibly when
$h$ is Boolean, whose tail is also squarefree. A Boolean leading
monomial cannot divide it. A divisor $Y_H\LM(g)$ from the rest of
$U$ would require $H\subseteq K$ and $\LM(g)\mid v$.
But $F_K\subseteq F_H$ gives $\I(F_H)\subseteq\I(F_K)$, hence
$\LM(g)$ belongs to the initial ideal of $\I(F_K)$. This contradicts
the fact that $v$ is a standard monomial for the reduced basis $H_K$.
For a retained $y_j^2-y_j$, its only tail $y_j$ could have a divisor
only if $y_j\in U$; that would also divide its leading monomial,
so it would not have been retained. Thus $V$ is monic and reduced,
and $V=G$. In fact activated $x$-Boolean polynomials with $K\ne\varnothing$
are always discarded, being divisible by their unactivated versions.

For truncation, enumerate only $|K|\le t$, construct $H_K$, and
keep $Y_Kh$ only if $|K|+\deg h\le t$; include the $y$-Boolean
polynomials if $t\ge2$. Call this list $U_t$. Any leading divisor
of a member of $U_t$ has degree at most $t$, since the order is
graded, and its activation set has size at most $t$. It is therefore
already in $U_t$. The same divisibility filter consequently returns
\emph{exactly} $G_t$, even if deriving an element by an unrelated
generating-set computation would require intermediate degree above $t$.

There are at most
\begin{equation}\label{eq:candidates}
 m+(2n+n^2+1)E_t(m)
\end{equation}
candidates before filtering, and $E_t(m)\le W_t$.
Each hard construction uses polynomially many 2-SAT tests on the
precomputed row clauses. A pairwise divisibility filter is polynomial
in this candidate count and $N+t$. Its arithmetic is on exponent
vectors and coefficients of constant magnitude. Finally, distinct
reduced leading monomials of degree at most $t$ give $|G_t|\le W_t$.
In sparse encoding the total length is
$O(W_t(t+2)\log(N+t+2))$, apart from fixed formatting conventions.
\end{proof}

For completeness, the connection with a clause-generated soft ideal
is exact. Let $c_{j,a}$ be the unary/binary forbidden-assignment
indicator for row $j$, or $1$ for an empty row relation. Then
\begin{equation}\label{eq:clauseideal}
 \I(S)=I_{\mathrm{bool}}+
             \langle y_jc_{j,a}(x):j,a\rangle.
\end{equation}
The zeros of the right side are exactly \eqref{eq:S} by
Lemma~\ref{lem:2decomp}; equality of ideals follows from the Boolean
product-of-fields argument. Equation~\eqref{eq:clauseideal} explains
why replacing a row by several clauses does not introduce independent
activation choices.

\begin{remark}[Why Boolean-only closure does not suffice]
The Boolean normal form of an arbitrary S-polynomial of hard family
members need not remain in those families. For example, with four
distinct variables, $q=uv$ and $r=(w-1)(v'-1)$ give
$S(q,r)=uvw+uvv'-uv$, a cubic already multilinear.
Both $q,r$ are quadratic conflicts. Their S-polynomial reduces to zero
using $q$, but Boolean reduction alone does not make it a family member.
Thus a closure assertion using Boolean reduction alone would be
false. Nor does an arbitrary cutoff of Buchberger intermediates prove
the correct reduced-basis truncation. Theorem~\ref{thm:structure}
instead works directly with the full vanishing ideal and proves the
structural conclusion, including shared row indicators.
\end{remark}

\section{Exact two-sided SoS derivations}
\label{sec:sos}
Here and below $\equiv$ denotes congruence modulo $I_{\mathrm{bool}}$.
We give explicit ideal multipliers at the end of the construction.

\subsection{Why scalar Farkas multipliers do not suffice}
The four Majority-closed rows
\[
 P_1=x+z-1,\quad P_2=x-z,\quad P_3=z-x,\quad P_4=1-x-z
\]
are Boolean-infeasible: the middle two force $x=z$, and the other two
then force $2x=1$. Their LP relaxation contains $(1/2,1/2)$, where
every row is zero. A nonnegative scalar combination of the rows cannot
be strictly negative on every Boolean point: its average on the four
points is its value at the centre, namely zero. Hence a uniform
scalar-conic refutation on the Boolean cube is impossible. Nevertheless,
\begin{equation}\label{eq:counterref}
 -1=(1-x)^2(P_1+P_2)+x^2(P_3+P_4)+4(x^2-x)
\end{equation}
is an exact degree-three SoS refutation. Polynomial square multipliers
and the Boolean equations are essential to this repair.

\subsection{Extracting clauses from a rational linear row}
Let $Q=b-\sum_i a_ix_i\ge0$ and let a forbidden assignment to distinct
variables $u,v$ have literal indicators $I_u,I_v$ (each $x_i$ or
$1-x_i$). Set $I=I_uI_v$, $J=(1-I_u)(1-I_v)$ and let
$-\delta<0$ be the maximum restricted slack, computed by
\eqref{eq:maxslack}. On the unfixed variables,
\[
 Q|_{u,v}=-\delta-\sum_{i\notin\{u,v\}}w_iL_i,
 \qquad w_i=|a_i|,\qquad
 L_i=\begin{cases}x_i&a_i\ge0,\\1-x_i&a_i<0.\end{cases}
\]
The clause inequality $C=1-I_u-I_v\ge0$ has the certificate
\begin{equation}\label{eq:extract}
 C\equiv J^2+\frac{I^2}{\delta}Q+
                 \sum_{i\notin\{u,v\}}\frac{w_i}{\delta}(IL_i)^2.
\end{equation}
Indeed, $I^2Q\equiv-\delta I-\sum_i w_iIL_i$, and
$(IL_i)^2\equiv IL_i$, so the right side reduces to $J-I=C$.
Every weight is a nonnegative rational; every multiplier displayed
as a square is a genuine square before reduction. The degrees of the
three types of terms are respectively $4,5,6$.
For a unary forbidden assignment with indicator $I$, the same calculation
gives
\begin{equation}\label{eq:unary}
 -I\equiv\frac{I^2}{\delta}Q+
          \sum_{i\text{ unfixed}}\frac{w_i}{\delta}(IL_i)^2
\end{equation}
in degree at most four. With no fixed variables and an infeasible row,
take $I=1$ to obtain a degree-two refutation. These formulas retain
the nonnegative slack terms missing from the scalar extraction claims.

\subsection{Implication paths and a degree-eight hard refutation}
\begin{lemma}\label{lem:hardref}
Every infeasible system of rational Majority-closed Boolean linear
rows has a constructible rational degree-eight SoS refutation. In
congruence form it is
\begin{equation}\label{eq:refute}
 -1\equiv R+\sum_j S_jP_j,
 \qquad\deg R\le8,\quad\deg S_j\le6,
\end{equation}
with $R,S_j$ sums of squares and polynomial bit complexity.
\end{lemma}
\begin{proof}
Extract the row clauses. An empty clause gives the preceding
degree-two refutation. Otherwise Lemma~\ref{lem:2sat} gives paths
$\neg x_v\leadsto x_v$ and $x_v\leadsto\neg x_v$. Interpret a
literal as its affine $0/1$ indicator. An edge $\ell\to\ell'$
has difference $D_e=\ell'-\ell$. For a binary clause this equals
its inequality $C_e$. For a unary clause with desired literal $\ell'$,
the inequality is $C_e=\ell'-1$ and $D_e=2C_e+1$. Thus
\[
 D_e=\theta_e C_e+c_e,
 \qquad(\theta_e,c_e)\in\{(1,0),(2,1)\}.
\]
This unary constant must not be omitted. Along the two paths the
differences telescope to $2x_v-1$ and $1-2x_v$, respectively.
The exact identity
\begin{equation}\label{eq:pathidentity}
 -1=(1-x_v)^2(2x_v-1)+x_v^2(1-2x_v)+4(x_v^2-x_v)
\end{equation}
then supplies the Boolean refutation.
Write \eqref{eq:extract} or \eqref{eq:unary} as
$C_e\equiv U_e+V_eP_{j(e)}$, where $U_e,V_e$ are SoS,
$\deg U_e\le6$, $\deg V_e\le4$. For each path occurrence let
$t_e=1-x_v$ on the first path and $t_e=x_v$ on the second. Set
\begin{equation}\label{eq:RS}
 R=\sum_e t_e^2(\theta_e U_e+c_e),\qquad
 S_j=\sum_{e:j(e)=j}\theta_e t_e^2V_e.
\end{equation}
These are SoS: products with $t_e^2$ multiply the individual square
factors. Substituting in \eqref{eq:pathidentity} proves
\eqref{eq:refute}; the row terms have degree at most seven and the
free SoS term at most eight. There are at most $4n-2$ edge occurrences
and $O(n)$ squares per extracted clause. Rational lengths are bounded
in Section~\ref{sec:bits}.
\end{proof}

\subsection{Conditioning and lifting by squared activation indicators}
\begin{theorem}[Two-sided simulation]\label{thm:simulation}
Suppose $p=Y_K\tau\in\I(S)$ with
$\tau\in\mathcal Q\cup\mathcal L\cup\mathcal Z$.
Both $p\ge0$ and $-p\ge0$ have explicit rational certificates from
the original soft system of degree at most
\begin{equation}\label{eq:Dexact}
 D_p=2|K|+2r+8,\qquad r=|\supp(\tau)|\le2.
\end{equation}
In particular, $D_p\le2\deg p+10$. Boolean-family elements have
immediate ideal certificates of their own degree. All these
certificates are constructible with polynomial bit complexity.
\end{theorem}
\begin{proof}
Lemma~\ref{lem:decomp} (or activating exactly $K$) gives
$\tau=0$ on $F_K$. Write $T=\supp(\tau)$, $|T|=r$. For each of
the at most four assignments $a\in\{0,1\}^T$ with
$v_a=\tau(a)\ne0$, the restricted system
$P_j^a=P_j|_{x_T=a}\ge0$, $j\in K$, is infeasible.
Fixing coordinates preserves Majority closure: the majority of three
extensions with the same fixed values still has those values.
Lemma~\ref{lem:hardref} therefore provides
\[
 -1\equiv R_a+\sum_{j\in K}S_{a,j}P_j^a,
 \qquad\deg R_a\le8,\quad\deg S_{a,j}\le6.
\]
These polynomials involve only the unfixed variables. Define
\[
 A_a=\prod_{i\in T}x_i^{a_i}(1-x_i)^{1-a_i},\qquad H_a=Y_K A_a.
\]
Multiply the refutation by the \emph{actual square} $H_a^2$.
Modulo Boolean equations, $H_a^2P_j^a\equiv H_a^2P_j$ and
$H_a^2(1-y_j)\equiv0$ for $j\in K$. Hence
\begin{equation}\label{eq:lift}
 -H_a\equiv H_a^2R_a+
                   \sum_{j\in K}H_a^2S_{a,j}\widehat P_j.
\end{equation}
For example the penalty cancellation is the exact ideal identity
\[
 Y_K^2(1-y_j)=-(y_j^2-y_j)y_j
                   \prod_{\ell\in K\setminus\{j\}}y_\ell^2.
\]
The fixing cancellation follows from $A_a^2(x_i-a_i)\equiv0$.
The explicit ideal-reduction formula below supplies both cancellations
as ordinary polynomial identities, without assuming any auxiliary
fixing equation in the final proof.

Boolean interpolation gives $p\equiv\sum_a v_aH_a$. For a chosen
sign $s\in\{1,-1\}$, set
\begin{align}
 \sigma_{s,0}&=\sum_{a:sv_a>0}sv_aH_a^2
       +\sum_{a:sv_a<0}(-sv_a)H_a^2R_a,\label{eq:sigma0}\\
 \sigma_{s,j}&=\sum_{a:sv_a<0}(-sv_a)H_a^2S_{a,j}
       \quad(j\in K),\label{eq:sigmaj}
\end{align}
and set all other row multipliers to zero. Every weight in these
expressions is positive. By \eqref{eq:lift},
\begin{equation}\label{eq:signcong}
 sp\equiv\sigma_{s,0}+\sum_{j\in K}\sigma_{s,j}\widehat P_j.
\end{equation}
Thus both signs have certificates using only the original inequalities.

The free term has degree at most $2|K|+2r+8$; each row term has
degree at most $2|K|+2r+7$. The positive interpolation terms have
degree at most $2|K|+2r$. Boolean ideal corrections do not increase
the largest of these degrees. Writing $q=\deg\tau$, the possibilities
are recorded in Table~\ref{tab:degrees}. Finally $Y_K(x_i^2-x_i)$,
and its negative, are direct Boolean ideal multiples; the same applies
to $y_j^2-y_j$.
\end{proof}

\begin{table}[htbp]
\centering
\begin{tabular}{@{}lccc@{}}
\toprule
Factor $\tau$ & $q$ & $r$ & Two-sided degree bound\\
\midrule
Quadratic conflict & 2 & 2 & $2\deg p+8$\\
Two-variable linear identity & 1 & 2 & $2\deg p+10$\\
Unary identity & 1 & 1 & $2\deg p+8$\\
Constant $1$ & 0 & 0 & $2\deg p+8$\\
Boolean factor & 2 & 1 & $\deg p$\\
\bottomrule
\end{tabular}
\caption{Formal degrees before Boolean reduction.}\label{tab:degrees}
\end{table}

\subsection{Explicit ideal multipliers and binary encoding}
\label{sec:bits}
For any ordinary polynomial $F=\sum_\alpha c_\alpha z^\alpha$ define
\begin{equation}\label{eq:Di}
 \mathcal D_i(F)=\sum_{\alpha:\alpha_i\ge2}c_\alpha
 \left(\prod_{h<i}z_h^{\min(\alpha_h,1)}\right)
 \left(\sum_{q=0}^{\alpha_i-2}z_i^q\right)
 \left(\prod_{h>i}z_h^{\alpha_h}\right).
\end{equation}
The telescoping identity
$z_i^a-z_i=(z_i^2-z_i)\sum_{q=0}^{a-2}z_i^q$ implies
\begin{equation}\label{eq:idealreduce}
 F=\ml(F)+\sum_i(z_i^2-z_i)\mathcal D_i(F),
 \qquad\deg((z_i^2-z_i)\mathcal D_i(F))\le\deg F.
\end{equation}
In particular put
$E_s=sp-\sigma_{s,0}-\sum_j\sigma_{s,j}\widehat P_j$.
Its multilinearization is zero by \eqref{eq:signcong}, so
\begin{equation}\label{eq:completeproof}
 sp=\sigma_{s,0}+\sum_j\sigma_{s,j}\widehat P_j
          +\sum_i(z_i^2-z_i)\mathcal D_i(E_s)
\end{equation}
is the exact certificate. This also makes every earlier congruence an
identity of the asserted degree. There are no residual auxiliary
clause axioms, hard-row axioms, or fixing axioms.

Here is a bit-complexity justification for the whole construction.
Clearing all input denominators by their product uses at most $L$
bits for that denominator. Each extraction gap $\delta$ is a sum
of signed input coefficients after fixing at most four coordinates
(two for $a$ and at most two for extraction). Its numerator and
denominator have $O(L+\log(N+2))$ bits; so does $1/\delta$ whenever
the gap is positive. A very small positive gap is therefore harmless
for bit complexity. The weights $w_i/\delta$ have the same bound.
There are at most four restricted refutations, each using at most
$4n-2$ path occurrences and $O(n)$ weighted squares per occurrence.
All other weights are $1,2$, the bounded values $|v_a|$, and the
original row coefficients, including $M_j$. A common denominator for
all expanded coefficients can be formed from polynomially many such
rationals, with polynomially many bits.

Multiplication by $Y_K^2$ shifts exponents and creates no new terms.
The $x$-part of each square in the constructed certificate is a product
of at most six literals, including the assignment indicator. Its
expansion has a uniformly bounded number of terms. Row multiplication
adds at most $n+2$ terms per product. Formula~\eqref{eq:Di} uses only
addition and monomial shifts, with at most $ND$ output contributions
per degree-$D$ monomial. Thus all operations and all coefficient bit
lengths are polynomial in $N+t+L$ for an individual structured
$p$ of degree at most $t$, in sparse encoding. In a dense Gram
encoding a safe uniform bound is $\poly(W_{2t+10},L)$ for construction,
total output length, and hence each coefficient's bit length.
The same bound covers simultaneous certification of $G_t$, since
$|G_t|\le W_t$. An effective magnitude bound $B$ is obtained by
constructing these rational certificates and taking the maximum of
one and the absolute values of all their entries. This gives a
computable $\log B=\poly(W_{2t+10},L)$, not an assumed coefficient bound.

Positive rational weights give rational diagonal Gram matrices on the
displayed square factors, hence rational PSD matrices on monomials.
If literal rational squares are desired, write a weight as $A/B>0$
with integers $A,B$. Expand $AB$ in binary. An even power $2^{2h}$
is one integer square and an odd power $2^{2h+1}$ is two such squares.
Dividing all square roots so obtained by $B$ expresses $A/B$ as a sum
of $O(\bits A+\bits B)$ rational squares. Multiplying these constants
by the displayed polynomial factors preserves degree and polynomial
encoding length.

\section{The general criteria in the form used here}
\label{sec:blackboxes}
The two black boxes below are the results entitled \emph{Dual criterion}
and \emph{Transfer from added identities to the original system} in
\cite[Theorems~3.2 and~3.3, respectively]{Mastrolilli2026}.
These numbers refer to the public version identified in the
bibliography. We specialize their notation and state the transfer
bound in the form needed here; we do not reproduce their proofs.
The signed-product identity used in the transfer is displayed
explicitly in~\eqref{eq:signed} below.
For any rational Boolean system let
$G_r$ denote the reduced graded vanishing basis through even degree $r$.

\begin{theorem}[Dual criterion, black box]\label{bb:dual}
Let $r=2d\ge2$. Suppose the coefficients of $G_r$ have magnitude
at most $C\ge1$, and both signs of each $b\in G_r$ have SoS
certificates of degree $D_b\le c d$, for a fixed constant $c$, with
Gram entries and equality coefficients bounded in magnitude by
$B\ge1$. Define
\begin{equation}\label{eq:Delta}
 \Delta=2\left\lceil\frac12\max\left\{r,
          \max_{b\in G_r}\bigl[D_b+2(r-\deg b)\bigr]\right\}\right\rceil.
\end{equation}
The maximum over an empty basis is omitted. If a rational polynomial
$f$ has a real degree-$r$ SoS proof from the original system, it has
a real degree-$\Delta$ Gram proof whose entries and equality
coefficients have magnitude
$2^{\poly(W_r,L+\bits f,\log C,\log B)}$.
Given computable bounds $B,C$ and an even search degree $O(d)$ at
least $\Delta$, a rational Gram certificate of $f+\eta\ge0$
is computable in polynomial time in the monomial count at that search
degree, $L$, $\bits f$, $\log B$, $\log C$, and $\bits\eta$,
for every rational $\eta>0$, under the stated existence promise.
The basis and its derivations need not be supplied to the search
algorithm. The perturbed certificate is an exact polynomial identity.
\end{theorem}

\begin{theorem}[Transfer, black box]\label{bb:transfer}
Let $\mathcal E$ be a finite list of nonzero rational polynomials of
degree at most even $r\ge2$, and suppose both signs of $b\in\mathcal E$
have original-system proofs of degree $D_b$. Replace $G_r$ by
$\mathcal E$ in \eqref{eq:Delta} to define $\Delta_{\mathcal E}$.
A degree-$r$ proof with the additional equality axioms $b=0$
transfers to the original system in degree $\Delta_{\mathcal E}$.
If its entries are bounded by $A\ge1$ and the two-sided proof entries
by $B\ge1$, the transferred entries are bounded by
\[
 A+2|\mathcal E|W_{\Delta_{\mathcal E}}^2(A+1)^2B.
\]
The substitution preserves rationality and has polynomial bit cost
in the encoded certificates and $W_{\Delta_{\mathcal E}}$.
On the moment side,
$\pi_r(S_{\Delta_{\mathcal E}})\subseteq S_r^{\mathcal E}$, where
the latter body imposes the added equalities with every allowed
multiplier. If their allowed multiples span all degree-$r$ vanishing
identities, then $S_r^{\mathcal E}=K_r$.
\end{theorem}

The degree calculation in this black box comes from the exact identity
\begin{equation}\label{eq:signed}
 ab=\left(\frac{a+1}{2}\right)^2b+
        \left(\frac{a-1}{2}\right)^2(-b).
\end{equation}
In our application the sharper per-element bound, rather than a single
worst-case degree substituted for every element, gives
\begin{equation}\label{eq:blowup}
 D_b+2(2d-\deg b)
       \le(2\deg b+10)+2(2d-\deg b)=4d+10.
\end{equation}
Boolean elements satisfy this bound as well. Since $4d+10$ is even,
both proof search and transfer may use exactly that degree as an upper
bound. Also $2\deg b+10\le4d+10\le14d$ for $d\ge1$, verifying the
constant-$c$ hypothesis. We do not assert a rational unperturbed proof
of $f$, or proof search at the original degree $2d$.

\section{Primal optimization with the computed basis}
\label{sec:primal}
We define the moment body explicitly to identify which relaxation is
optimized. For even $r\ge2$, put
$\mathcal J_r=\{A\subseteq[N]:|A|\le r\}$ and let
$u=(u_A)_{A\in\mathcal J_r}$ represent the functional
$L_u(z_A)=u_A$, extended using Boolean multilinearization.
Let $S_r$ be defined by normalization $u_\varnothing=1$ and
\begin{align}
 (u_{A\cup B})_{A,B\in\mathcal J_{r/2}}&\succeq0,
       \label{eq:moment}\\
 (L_u(\widehat P_jz_Az_B))_{A,B\in\mathcal J_{a_j}}&\succeq0,
 \qquad a_j=\left\lfloor\frac{r-\deg\widehat P_j}{2}\right\rfloor,
       \label{eq:localizing}
\end{align}
for all nonzero input rows. Boolean equality multiples hold by the
definition of $L_u$. The ordinary degree of the row is used in
\eqref{eq:localizing}; normally it is one, but constant rows are handled
by the same formula. Define
\begin{equation}\label{eq:K}
 K_r=\{u\in S_r:L_u(h)=0\text{ for all }
                  h\in\I(S)\cap\Q[z]_{\le r}\}.
\end{equation}
This is the strengthened body, independently of which complete
identity list represents it.

The following is the \emph{Supplied-Gr\"obner-basis primal optimization}
theorem of~\cite[Theorem~4.3]{Mastrolilli2026}, specialized to the
notation of the present paper.
\begin{theorem}[Supplied-Gr\"obner-basis primal optimization, black box]
\label{bb:primal}
For a rational Boolean system with nonempty feasible set $S$, suppose
the correct reduced graded lexicographic basis truncation $G_D$ of
$\I(S)$ is supplied, where $D\ge r=2d\ge2$, with rational encoding
length $L_G$. Filter it to $G_r$. The list
\[
 Q_r^{\mathrm{mult}}=\{z^\alpha b:b\in G_r,\
                              |\alpha|+\deg b\le r\}
\]
spans $\I(S)\cap\Q[z]_{\le r}$. Graded normal forms alternatively
give a complete list and a rational affine parametrization by standard
monomial moments. For every rational $f$ of degree at most $r$ and
rational $\eta>0$, a rational $u^*\in K_r$ with
\[
 L_{u^*}(f)\ge\max_{u\in K_r}L_u(f)-\eta
\]
is computable in time and output length polynomial in
$W_r,L,L_G,\bits f,\bits\eta$. All normalization, equality,
and PSD constraints are satisfied exactly. The moment degree remains
$r$ even when $D>r$.
\end{theorem}

Here the promise of a correct supplied basis is discharged by
Theorem~\ref{thm:structure}: we compute $G_r$ and take $D=r$.
We describe the algebraic input and affine coordinates explicitly.

\begin{proposition}[Complete identities and affine coordinates]
\label{prop:coords}
For the soft system, a complete rational basis of the vector space
$\I(S)\cap\Q[z]_{\le r}$ and its normal-form table are computable
in $\poly(W_r,L)$ bit time. If the standard monomials through degree
$r$ are $\{1,v_1,\ldots,v_p\}$ and
\begin{equation}\label{eq:NF}
 \NF(z_A)=a_{A,0}+\sum_{i=1}^p a_{A,i}v_i,
\end{equation}
then
\begin{equation}\label{eq:Phi}
 \Phi(w)_A=a_{A,0}+\sum_{i=1}^p a_{A,i}w_i
\end{equation}
is a rational affine bijection from $\R^p$ onto
\[
 \mathcal H_r=\aff\{e_r(s):s\in S\},\qquad
 e_r(s)=(z_A(s))_{A\in\mathcal J_r}.
\]
Its free coordinates are original moments: if $z_A=v_i$, then
$\Phi(w)_A=w_i$. Moreover $\aff K_r=\mathcal H_r$.
\end{proposition}
\begin{proof}
Divide each ordinary monomial of degree at most $r$ by $G_r$.
Graded division never increases degree and uses at most $W_r$
cancellation steps, because each step removes the largest remaining
nonstandard monomial. The remainder is the unique reduced normal
form. Coefficients of $G_r$ have magnitude at most one and at most
four terms, so a coarse coefficient $\ell_1$ bound during division
is $4^{W_r}$ for a unit monomial input. In particular the rational
table has polynomial bit length; in fact its coefficients are integers.

The polynomials $z^\alpha-\NF(z^\alpha)$ for nonstandard monomials
of degree at most $r$ form a \emph{linearly independent} complete
identity list: each has its own nonstandard monomial with coefficient
one and otherwise only standard monomials. For any identity $h$,
linearity and $\NF(h)=0$ express $h$ as the corresponding linear
combination. This proves both completeness and the claimed vector-space
basis, including ordinary Boolean ideal multiples.

Every standard monomial is squarefree, since Boolean square leading
monomials lie in the initial ideal. The constant is standard since
$S\ne\varnothing$. Equations~\eqref{eq:NF} therefore give
\eqref{eq:Phi}, with free original coordinates. Any assignment of
values to the standard monomials, and value one to the constant,
extends by normal forms to a functional annihilating all truncated
identities. Conversely those identities force \eqref{eq:Phi}.
To identify this space with the affine hull of evaluations, consider
the Boolean evaluation matrix on $S$ for the proof only. Its kernel
is precisely the truncated multilinear identity space; its rational
and real ranks agree. Annihilating that space places the moment
vector in the real span of the evaluations. Normalization makes the
coefficients sum to one, so the span becomes their affine hull.
Finally,
\[
 \conv\{e_r(s):s\in S\}\subseteq K_r\subseteq\mathcal H_r
\]
implies $\aff K_r=\mathcal H_r$. No feasible-set enumeration is used
by the algorithm.
\end{proof}

Exact rational feasibility is the output guarantee of
Theorem~4.3 of~\cite{Mastrolilli2026}, invoked here as a black box; it
is not inferred merely from approximate floating-point SDP feasibility.
The normal-form coordinates in \eqref{eq:Phi} identify the affine hull
required by that theorem. We do not reproduce its geometric optimization
argument.

\subsection{Main primal-dual tractability theorem}
\label{sec:main}
The structural and proof-theoretic results above combine with the
general criteria to give the following computational theorem. The basis
construction supplies the primal algorithm, while the two-sided
derivations supply the dual hypothesis and the comparison with the
higher level of the original hierarchy.
\begin{theorem}[Primal-dual tractability of soft Majority systems]
\label{thm:main}
Let \eqref{eq:soft} be a rational soft system with each original row
Majority-closed and $M_j$ satisfying \eqref{eq:M}. Set $N=n+m$ and
let $L$ be its input length. For every integer $d\ge1$, put
\[
 r=2d,\qquad \kappa(d)=4d+10,\qquad
 W_s=\binom{N+s}{s}.
\]
The following assertions hold without a supplied structural promise
beyond the rowwise Majority assumption.
\begin{enumerate}[label=\textup{(\roman*)}]
\item The reduced graded lexicographic truncation $G_r$ is computable
in $\poly(W_r,L)$ bit time. It has at most $W_r$ members, at most
four terms per member, and coefficients in $\{0,1,-1\}$. Its
non-Boolean members are $Y_K\tau$, with $\tau$ on at most two
original variables. Both signs of each member $b$ have rational
SoS certificates of degree at most $2\deg b+10$. Their total dense
encoding and construction time are $\poly(W_{\kappa(d)},L)$.
\item Given rational $f$ of degree at most $2d$ and rational
$\eta>0$, under the promise that $f\ge0$ has a real
degree-$2d$ SoS proof from the original soft system, a rational
Gram certificate of $f+\eta\ge0$ from that same system is
computable in degree at most $4d+10$. Its construction time and
total bit length are
\[
 \poly(W_{4d+10},L,\bits f,\bits\eta).
\]
\item For every rational objective $f$ of degree at most $2d$ and
rational $\eta>0$, one can compute a rational
pseudoexpectation $L_{u^*}$ through degree $2d$ that is exactly
feasible in $K_{2d}$ and satisfies
\begin{equation}\label{eq:optimal}
 \max_{u\in K_{2d}}L_u(f)-\eta
       \le L_{u^*}(f)\le\max_{u\in K_{2d}}L_u(f).
\end{equation}
Its computation time and total output length are
$\poly(W_{2d},L,\bits f,\bits\eta)$.
\item The original hierarchy and the strengthened body satisfy
\begin{equation}\label{eq:sandwich}
 \pi_{2d}(S_{4d+10})\subseteq K_{2d}\subseteq S_{2d}.
\end{equation}
For maximization this gives
\begin{equation}\label{eq:values}
 \max_{s\in S}f(s)
 \le\max_{u\in S_{4d+10}}L_u(f)
 \le\max_{u\in K_{2d}}L_u(f)
 \le\max_{u\in S_{2d}}L_u(f).
\end{equation}
In particular $L_{u^*}(f)\ge\max_{u\in S_{4d+10}}L_u(f)-\eta$.
For minimization, the inclusions give the reversed chain of minima,
and applying \textup{(iii)} to $-f$ returns an $\eta$-optimal
feasible minimizer.
\end{enumerate}
For fixed $d$ these are polynomial-time Turing algorithms. For
$\eta=2^{-s}$, all precision dependence is polynomial in $s+1$.
Numerical coefficients of the input need not be polynomially bounded;
only their binary encoding length enters the bounds.
\end{theorem}
\begin{proof}
Part (i) combines Theorems~\ref{thm:structure} and
\ref{thm:simulation}, using $2r+10=4d+10$ for simultaneous
certificate construction. Their explicitly computed magnitudes give
the bound $B$ in Theorem~\ref{bb:dual}, and $C=1$.
Equation~\eqref{eq:blowup} and that black box prove (ii).
For (iii), compute $G_{2d}$ and supply it to
Theorem~\ref{bb:primal}; its encoding length is already polynomial
in $W_{2d},L$. Nonemptiness is guaranteed by $y=0$.

For (iv), any identity $h$ of degree at most $2d$ divides as
$h=\sum_{b\in G_{2d}}a_bb$ with
$\deg a_b\le2d-\deg b$. By \eqref{eq:signed} and the verified
per-element degrees, both signs of each $a_bb$, and hence of $h$,
have proofs of degree at most $4d+10$ from the original system.
SoS soundness forces $L_u(h)=0$ for every $u\in S_{4d+10}$.
Restriction preserves the lower-degree PSD constraints and normalization,
giving the first inclusion in \eqref{eq:sandwich}; the second follows
from \eqref{eq:K}. This is also the application of
Theorem~\ref{bb:transfer} with $\mathcal E=G_{2d}$ and its complete
list of allowed multiples. Every feasible evaluation belongs to every
level. Boolean moment bounds make these bodies compact, so their
optima exist and inclusions yield \eqref{eq:values}. Combining it
with \eqref{eq:optimal} gives the last assertion.
\end{proof}

\begin{remark}[Degrees and limits of the conclusion]
The algebraic degree is $D=2d$; a correct supplied $G_D$ with $D\ge2d$
can instead be filtered, with cost charged to its encoding.
The optimized moments have degree $2d$ (matrix order $d$); basis
derivations have degree at most $2\deg b+10$; eliminating arbitrary
permitted equality multipliers costs at most $4d+10$.
Neither $K_{2d}=S_{2d}$ nor extendibility of $u^*$ to $S_{4d+10}$ is
asserted. Its objective need not approximate the higher-level optimum;
only the one-sided comparison in part (iv) is guaranteed.
Proof search perturbs $f$ by the scalar $\eta$, while primal
feasibility is exact.
\end{remark}

\clearpage
\phantomsection
\addcontentsline{toc}{section}{References}
\begingroup
\small
\bibliographystyle{plain}
\bibliography{robust_majority}
\endgroup
\end{document}